\documentclass[lettersize,journal]{IEEEtran}

\usepackage[utf8]{inputenc}

\usepackage{cite}
\usepackage{amsmath,amssymb,amsfonts}
\usepackage{ragged2e}
\usepackage{blindtext}
\usepackage{graphicx}
\usepackage{textcomp}
\usepackage[dvipsnames]{xcolor}
\usepackage{soul}
\usepackage{url}
\usepackage{fancyhdr, lipsum}
\usepackage{setspace}

\usepackage{enumitem}
\usepackage[symbol]{footmisc}
\def\BibTeX{{\rm B\kern-.05em{\sc i\kern-.025em b}\kern-.08em
		T\kern-.1667em\lower.7ex\hbox{E}\kern-.125emX}}

\usepackage[utf8]{inputenc}
\usepackage[english]{babel}
\usepackage{amsthm}
\newtheorem{theorem}{Theorem}

\usepackage{verbatim}
\usepackage{tikz}
\usepackage{pgfplots}
\usepackage[normalem]{ulem}
\usepackage{algcompatible}
\usepackage{algpseudocode}

\usepackage{xcolor}

\algnewcommand{\LineComment}[1]{\Statex \hfill \(\triangleright\) #1}

\definecolor{myred}{RGB}{255,189,200}
\definecolor{myblue}{RGB}{200,228,250}

\usepackage{pifont}

\pgfmathsetmacro{\nodebasesize}{1} 
\pgfmathsetmacro{\nodeinnersep}{0.05}

\usepackage{xcolor}
\usepackage{threeparttable}
\usepackage{pifont}
\usepackage{etaremune}
\usepackage{caption}
\usepackage{subcaption}
\usepackage[export]{adjustbox}
\usepackage{listings}
\usepackage{courier}

\usepackage{algorithm} 
\usepackage{algorithmicx}
\usepackage{fixltx2e}
\usepackage{comment}
\usepackage{multirow, booktabs}
\usepackage{siunitx}
\usepackage{dblfloatfix}
\usepackage{xargs}  

\usepackage{tikz}
\usetikzlibrary{matrix}
\usetikzlibrary{calc,fit}

\setlist{leftmargin=4.5mm}

\usepackage{tikz}
\usepackage{calc}

\begin{document}

\title{HiKonv: Maximizing the Throughput of Quantized Convolution With Novel Bit-wise Management and Computation}


\author{Yao Chen,
        Junhao Pan,
        Xinheng Liu,
        Jinjun Xiong,~\IEEEmembership{Fellow,~IEEE},
        and~Deming Chen,~\IEEEmembership{Fellow,~IEEE}%
\IEEEcompsocitemizethanks{
\IEEEcompsocthanksitem Yao Chen is with the Department of Computer Science, National University of Singapore, Singapore (email: yaochen@nus.edu.sg).
\IEEEcompsocthanksitem Junhao Pan and Xinheng Liu are with the Department of Electrical and Computer Engineering, University of Illinois at Urbana-Champaign, Urbana, IL, USA (email: \{jpan22, xliu79\}@illinois.edu).
\IEEEcompsocthanksitem Jinjun Xiong is with the School of Engineering and Applied Sciences, University at Buffalo, Buffalo, NY, USA (email: jinjun@buffalo.edu).
\IEEEcompsocthanksitem Deming Chen is with the Department of Electrical and Computer Engineering, University of Illinois at Urbana-Champaign, Urbana, IL, USA (email: dchen@illinois.edu).}
\thanks{Manuscript received xxx xxx, xxx.}
}



\maketitle

\begin{abstract}
Quantization is proven to be effective for Convolutional Neural Networks (CNN) to reduce the cost of computation and storage with low-bitwidth data representations. 
However, the current execution of quantized data on the existing full-bitwidth processing units, such as ALU in CPUs and DSP in FPGAs, is through simply extending the lower bitwidth to the supported bitwidth, which leads to the underutilization of the computing unit and delivers low computational throughput.
In this study, we propose HiKonv, a unified solution that maximizes the throughput of convolution on a given underlying processing unit with low-bitwidth quantized data as inputs through novel bit-wise management and parallel computation.
We establish theoretical framework and performance models using a full-bitwidth multiplier for highly parallelized low-bitwidth convolution and demonstrate new breakthroughs for high-performance computing in this critical domain. For example, a single 32-bit processing unit in CPU can deliver 128 binarized convolution operations (multiplications and additions), thirteen 4-bit convolution operations or five 8-bit convolution operations with a single multiplication instruction, and a single 27$\times$18 multiplier in the FPGA DSP core can deliver 60, 8 or 2 convolution operations with 1, 4 or 8-bit inputs in one clock cycle.
We demonstrate the effectiveness of HiKonv on CPU and FPGA for both convolutional layers and a complete DNN model with our platform-oriented implementations.
On CPU, HiKonv outperforms the baseline implementation with $1$ to $8$-bit inputs and provides up to $7.6\times$ and $1.4\times$ performance improvements for 1-D convolution.
For the 2-D convolutional layer, HiKonv performs $2.74\times$ and $3.19\times$ over the baseline implementation for 4-bit signed and unsigned data inputs.
HiKonv also provides over $2\times$ latency improvement for a complete DNN model on both Intel and ARM CPUs.
On FPGA, the HiKonv solution enables a single DSP to process the same convolution operations that require multiple DSPs in the conventional convolution method with a shorter processing latency.
For binarized input, each DSP with HiKonv is equivalent up to 76.6 LUTs.
Compared to the DAC-SDC 2020 champion model for FPGA, HiKonv achieves a 2.37$\times$ throughput improvement and 2.61$\times$ DSP efficiency improvement, respectively.
\end{abstract}


\section{Introduction}\label{sec:intro}

Quantization is typically done by approximating high-precision floating point values to low-bitwidth integers or fixed-point representations.
It is commonly used in the deployment of Deep Neural Network (DNN) models to reduce the cost (i.e., memory consumption) and improve performance (i.e., execution time)~\cite{gholami2021survey, ul2q, vecq, skynet, DNNBuilder, cong2019dac}.
This is particularly important for modern DNN models, as many of them employ convolutional layers, which contain intensive multiplication and accumulation (MAC)
operations~\cite{chen2019clouddnn, DNNBuilder, cong2019dac,tdla, csrnet}.
Therefore, many novel quantization methods have been proposed in the literature to reduce the precision of weights, activations, and even gradients to low-bitwidth integers or fixed-point types for DNNs while retaining high model accuracy~\cite{ul2q, vecq, cong2019dac, zhou2016dorefa, tdla}.

Despite the success of quantization for DNNs, the current deployment of quantized convolution to the computing platform is not ideal.
To take advantage of the lower bitwidth of data from quantization, existing solutions tend to design new computing architectures that are specific to quantized convolutions~\cite{sharma2018bit,ryu2019bitblade,shin2018dnpu,lee2018unpu,sharify2018loom,pirdsp}, which lose the generality and take long application-specific integrated circuit (ASIC) design cycles.
While for computing platforms with given fixed arithmetic units, there is no general support for quantized convolution operation. Most general processing units such as X86\_64 CPUs and ARM processors have a high-bitwidth (such as 32 or 64 bits) multiplication-and-accumulation (MAC) unit for either floating-point numbers or integers~\cite{mliot, AVXforquant}.
When they are used for quantized MACs, most of the bitwidths are left underutilized, wasting precious computing resources~\cite{xilinxint4, xilinxint8, AVXforquant, bertTacl}.
Even with the 8-bit multi-data processing of the Advanced Vector Extensions (AVX) support in X86\_64 architecture,
processing a single 4-bit multiplication would still occupy the 8-bit data channel with
the remaining 4 bits simply wasted~\cite{AVXforquant}.
Waste becomes even more severe when either processing lower bitwidth (such as binary and ternary) data or utilizing a hardware unit with higher built-in bitwidths support.
Reconfigurable hardware, such as FPGA, can alleviate some of the waste because of its bit-level flexibility for hardware configuration, but it exhibits similar drawbacks, especially when adopting the high-precision digital signal processing (DSP) units in the FPGAs~\cite{mliot, tdla, chen2016platform, realtime, fcudanoc, ChenPlatformchoice}.
Without careful bit-wise management of inputs and outputs, deploying quantized DNNs on FPGAs with the given DSPs still wastes a lot of their computation capacity~\cite{xilinxint8, xilinxint4, chen2019clouddnn, fcudanoc}.

The lack of exploration of the relation of the low-level computing pattern of existing MAC or multiplication units to convolution constrains further improvement of the processing performance of the given arithmetic units for quantized convolution. 
In this paper, we propose a novel solution, HiKonv, that maximizes the efficiency of the existing multiplication units when conducting convolution operations with quantized inputs, thus improving the throughput of quantized convolution on the arithmetic unit and reducing the end-to-end DNN computation latency.
The major contributions of our HiKonv solution are as follows:
\begin{itemize}
\item We first conduct theoretical exploration to show that HiKonv's bit-wise management is universal, and it adopts a single high-bitwidth multiplier for multiple quantized convolutions in a single multiplication operation. This technique can be applied to arbitrarily quantized bitwidths with a given high-bitwidth multiplier unit.
\item We then extend it to support arbitrary lengths of 1-D and 2-D convolutions with the corresponding bit management and computation.
\item We build the theoretical analysis model that is used to explore the optimal design configurations with the given quantization bitwidths and the arithmetic unit. 
\item We provide two sets of implementations of HiKonv on general purpose processor and FPGA DSP, respectively, with different design considerations.
\item Our experimental results further validate the general applicability of HiKonv, the effectiveness of the analytical model, and the performance improvement of our implementations. For example, our CPU-based implementation of HiKonv achieves up to $2.4\times$ and $2.03\times$ latency improvement for a 4-bit quantized DNN model on X86\_64 CPU and ARM processor, respectively. The latency of our FPGA-based implementation of the DNN model outperforms the state-of-the-art implementation by $2.37\times$.

\end{itemize}

Because of its generality, we believe that HiKonv opens up a new venue for further improving the hardware efficiency of DNN-based inferences. It
not only improves the throughput and latency for existing quantized DNN models on existing hardware, but also offers new opportunities
for designing new hardware-friendly quantized DNN models or co-designing both the hardware and quantized DNN models~\cite{edd, DNNBuilder, Hao2019FPGADNNCA, codesignoandc}.

The rest of the paper is organized as follows. Section~\ref{sec:backnrelate} surveys existing solutions for the processing of quantized convolutions. Section~\ref{sec:prelim} introduces the preliminary information of our HiKonv solution, while Section~\ref{sec:formula} presents the theoretical deduction of the HiKonv solution together with the detailed throughput analysis for optimal design configuration. Section~\ref{sec:impl} shows the implementations of HiKonv on general-purpose processors and FPGAs with different design considerations. Section~\ref{sec:eva} presents the evaluation results. Finally, Section~\ref{sec:conclusion} concludes this paper.

\section{Related Works}\label{sec:backnrelate}

\begin{figure*}[]
    \centering
    \includegraphics[width=0.7\textwidth]{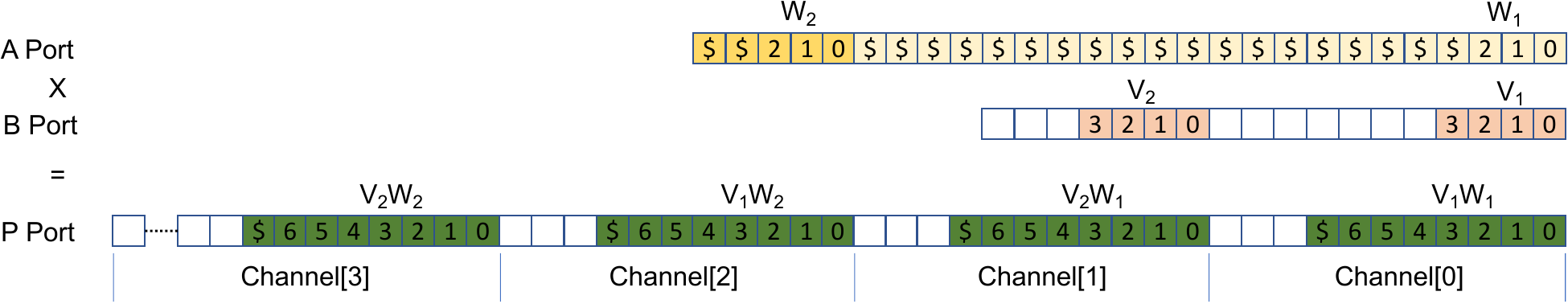}
    \caption{INT4 optimization on DSP48E2~\cite{xilinxint4}.}
    \label{fig:xilinxint4}
\end{figure*}

The existing solutions for quantized convolution processing rely on the quantized multiply and accumulation support of the underlying processing units.
These solutions can be classified into two categories: 1) solutions with dedicated hardware architecture for quantized/low-precision arithmetic and 2) solutions that rely on optimized software libraries to utilize the existing processing units more efficiently.
Note that although modern General Purpose Processors (GPP) equipped with Single Instruction Multiple Data (SIMD) floating point processing units can also support certain low-bitwidth data processing, they belong to the software library category because the processing units can also be used for other general arithmetic operations, i.e., AVX2 in X86\_64 CPUs and SSE in ARM processors.
To support generality, in this work, we will only focus on the main computing structure, such as the Arithmetic and Logic unit (ALU) or the multiplier in the CPU or FPGA DSP, instead of dedicated hardware architectures.

\subsection{Optimized software library} 

There are methods that pack/insert shorter bitwidth inputs into longer words and attempt to use the existing high-bitwidth computation units to improve the processing efficiency for quantized inputs~\cite{AVXforquant, xilinxint4, ottavi2020mixed,garofalo2020pulp,lai2018cmsis, xilinxint8}.
These solutions are designed as software libraries to fully utilize the existing hardware units such as ALU and Floating Point Unit (FPU) in the X86\_64 CPUs or ARM CPUs, etc. 
Here we consider the methods that do not involve the change of the existing hardware.

A library called \textit{Clover}~\cite{AVXforquant} adopts the AVX support of X86\_64 architecture and further extends it with an optimized 4-bit data container that reduces the memory access overhead. It simply adopts the 8-bit channel supported by the AVX architecture to process 4-bit data, with an online 4-bit to 8-bit extension. Although the performance of the system is improved, the speed-up is constrained by the number of supported channels by the AVX architecture.

CMSIS-NN~\cite{lai2018cmsis} is a set of software kernels for the execution of neural networks (NNs) on the ARM cortex-M CPUs, where the 16-bit single instruction multiple data (SIMD) MAC unit is used to process the low-bitwidth input data with a low-bitwidth to 16-bit conversion, i.e., 8-bit to 16-bit conversion. As the underlying MAC units only support 16-bit data inputs, all data with bitwidth lower than 16 is converted to 16-bit. Similar to \textit{Clover}, the performance gain comes from efficient bitwidth conversion and full utilization of the 16-bit MAC channels.

PULP-NN\cite{garofalo2020pulp} is a software library specifically designed for GAP-8\cite{GAP8} processors. Each such processor contains 8 RISC-V cores with special DSP extensions. The execution of PULP-NN requires unique bit-wise operations, i.e., low-bitwidth data unpacking, data quantization instructions, etc., that are supported by the processor. It supports sub-word size data to reduce the memory access cost and fully utilizes the SIMD instructions to process data in parallel. It shows great efficiency due to the special instruction set, which also limits its flexibility.

\subsection{Leveraging Reconfigurable Hardware}

XILINX INT4 and INT8~{\cite{xilinxint4, xilinxint8}} support specifically for 4-bit and 8-bit processing on Xilinx FPGAs. 
They pack multiple inputs to enable multi-channel multiplication
and can be easily extended to process lower bitwidths than 4 and 8 by bitwidth extension. Particularly, both methods adopt existing $28\times17$ DSP units and pack the low-bitwidth data to the input ports to process two or four multiplication operations simultaneously, which improves the processing efficiency of the DSP.  Although these solutions require the data to be either packed before sending to the DSP or use dedicated logic paths to pack the data during runtime, 
they use existing DSP units, and the solutions can be simply wrapped as High Level Synthesis (HLS) libraries.

\subsection{Motivation}

Developing dedicated hardware solutions requires long design cycles and also loses the flexibility to handle changing data types if the accelerator's data path is fixed to a specific bitwidth.
Software solutions on existing structures improve processing efficiency for quantized operations only by fully utilizing memory bandwidth or increasing the concurrency by a small factor, e.g., 2 or 4~\cite{xilinxint8, xilinxint4}.
\textbf{Furthermore, none of the existing solutions takes into account the internal relation between multiplication and convolution and the potential to further improve the throughput of quantized convolutions}.
There are no theoretical studies to guide the flexible management of low-bitwidth convolution computations. Our work, HiKonv, fills these existing gaps and provides theoretical guidance for the best computational efficiency and throughput on either existing hardware architecture or on new bit-efficient processing units for flexible low-bitwidth convolution computation.
\section{Preliminary}\label{sec:prelim}

Our proposed HiKonv solution is inspired by the input packing for multi-channel multiplication on FPGA DSP~\cite{xilinxint4, xilinxint8}.
It adopts similar definitions as the existing input packing methods with proper extensions and revisions. In this section, we first look into the existing input packing method and then revisit the detailed operational pattern of a 1-D convolution.

\subsection{Single Multiplier Multi-Channel Multiplication}

Compared to simply expanding the low-bitwidth data to a higher bitwidth alternative, typically 16 or 32 bits, the multi-channel multiplication is
a representative method used by current hardware units to process multiple low-bitwidth inputs concurrently~\cite{xilinxint4, xilinxint8}.
Particularly, multiple low-bitwidth values are packed together to form larger bitwidth multiplicands.
By taking advantage of the shift-addition operation of the intermediate results during the process,
\textcolor{black}{the output of the high-bitwidth multiplications, $Prod$, can be separated into multiple individual channels holding results of low-bitwidth multiplications}.
The most representative example of multi-channel multiplication is the INT4 optimization for the Xilinx DSP48E2 unit, as shown in Figure~\ref{fig:xilinxint4}. 
\textcolor{black}{Each of the inputs contains two 4-bit integer values  and form the inputs to the input ports $A$ and $B$}.

The separation of the four different output channel in Figure~\ref{fig:xilinxint4}
provides the opportunity to simply segment out the required multiplication results $Prod$ from the output port $P$.
To avoid the overlapping of the results, some bits are reserved or filled with values during the packing of the low-bitwidth data.
Note here, the \{$V$\} values are all unsigned integers (UINT4), and \{${W}$\} values are signed integers (INT4).
An INT4$\times$UINT4 multiplication generates a result with at least an 8-bit space.
Here, the existing methods define the term \textit{guard bit} as the filling 1s or 0s between the output results for the purpose of preventing the overflow of the multiplication.
In this example, the value of \textit{guard bit} is 3.
The multiplication with four low-bitwidth inputs for this example is represented as:
\textcolor{black}{
\begin{equation}
\label{equ:xilinx4bit}
\small
\begin{split}
    Prod &= A \times B \\
    &= (V_2 \cdot 2^{11} + V_1) \cdot (W_2 \cdot 2^{22} + W_1) \\
    &= V_2W_2\cdot 2^{33} + V_1W_2\cdot2^{22} + V_2W_1\cdot2^{11} + V_1W_1
\end{split}    
\end{equation}
}The output of Equation~\ref{equ:xilinx4bit} is the concatenation of four multiplication results with zeros between them due to the guard bits. This process accomplishes four channels of multiplication in one operation cycle with a large bitwidth multiplier.

\subsection{1-D Convolution}
Denote $F_{N,K}(f,g)$ as
the conventional 1-D discrete convolution between an $N$-element sequence $f$ and a $K$-element kernel $g$. Here, we define the infinite-length sequence $h$ as the extension of $f$ with zero values to the index range of $(-\infty,\infty)$, as shown in Equation~\ref{eq:pad_f0}. 
The 1-D convolution is represented in Equation~\ref{eq:pad_f}.
\textcolor{black}{Meanwhile, $y'$ is the output with $N\text{+}K\text{-}1$ non-zero elements}.

\vspace{-4mm}
\begin{small}
\begin{gather}
\label{eq:pad_f0}
    h[n]=\begin{cases} 
        f[n] &,  0\leq n < N\\
        0 &, n <0 \text{ or } n\geq N
    \end{cases}\\
    y'[m]= (h\ast g)[m]=\sum_{k=0}^{K-1}h[m-k]g[k] \label{eq:pad_f}
\end{gather}
\end{small}
\textcolor{black}{Alternatively, the non-zero sequence in $y'$ can be represented with Equation~\ref{eq:ym} as the sum of an $(N\text{+}K\text{-}1)$-element sequence. Each of the $y[m]$ involves a sequence of multiplication and addition operations of continuous elements in $h$ and $g$, and we denote it as \textit{partial convolution}.}
\begin{equation}
\small
\label{eq:ym}
y[m]=\sum_{k+n=m} h[n]g[k]
\end{equation}
\section{Multiplier for Convolution}\label{sec:formula}

\begin{figure*}
    \centering
    \includegraphics[width=\textwidth]{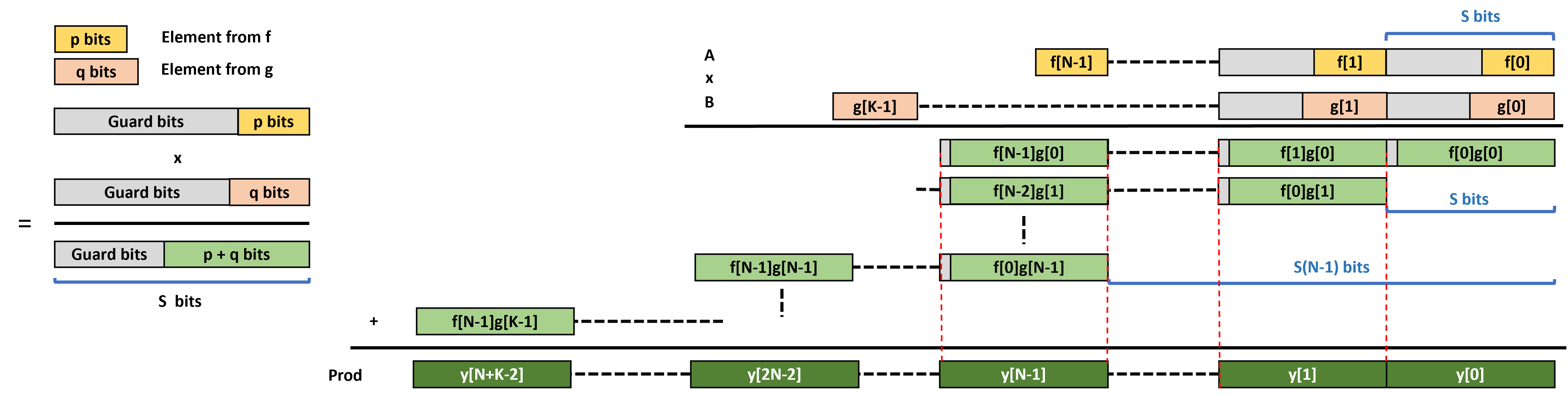}
    \caption{Binary view of the ideal process of $Prod = A\times B$.}
    \label{fig:amulb}
\end{figure*}

There are two observations from the single multiplier multi-channel multiplication and 1-D convolution:
\begin{itemize}
    \item First, multi-channel multiplication creates a set of multiplications and accumulations of low-bitwidth inputs within the operation of a single large bitwidth multiplication.
    \item Second, 1-D convolution is constructed with multiplications and accumulations of a set of contiguous data from two different data sequences.
\end{itemize}
Inspired by these two observations, we further extend multi-channel multiplication 
with a novel bit-wise management and generalize the solution for using a given hardware unit to process the maximum amount of low-bitwidth convolution operations concurrently. Furthermore, we provide theoretical guarantees for our solution.

First, we define the variables related to our exploration.
As shown in Figure~\ref{fig:amulb}, we assume a given high-precision hardware unit that can multiply $L_A$-bit integer input $A$ with $L_B$-bit integer input $B$ and generate the product $Prod$. The bitwidths of $A$ and $B$ define the computation capability of the hardware unit, or more specifically, the multiplier. Convolution input $f$ and kernel $g$ are the two sequences of low-bitwidth integer values quantized to $p$ and $q$ bits, respectively. 
Note here, Figure~\ref{fig:amulb} is for the case where all elements of the sequences $f$ and $g$ are unsigned integers to ease the presentation of the HiKonv solution, and later we will show the generality of HiKonv for both signed and unsigned values as inputs.

To determine how to load $A$ and $B$ with multiple convolution operands from $f$ and $g$ and perform the convolution between these operands, we first define the concept of a slice, which is a certain length of consecutive bits in the input to hold the low-bitwidth input data, $S$ is the size of a slice. Both input $A$ and $B$ are segmented with slices, as demonstrated on the left in Figure~\ref{fig:amulb}.
The lower bits of each slice contains one element from $f$ or $g$.
To simplify the problem, we assume that $N$ and $K$ are the maximum numbers of elements from $f$ and $g$ that fit into $A$ and $B$, respectively.
Hence, the polynomial representations of $A$ and $B$ are:
\begin{small}
\begin{equation}\label{eq:inpack}
\begin{split}
    & A=\sum_{n=0}^{N-1} f[n]\cdot 2^{ S\cdot n } \\
    & B=\sum_{k=0}^{K-1} g[k]\cdot 2^{ S\cdot k }
\end{split}
\end{equation}
\end{small}

Although the intermediate results of the multiplication are invisible to us, we assume that the processing unit takes the most ideal way of shift-add operations for the multiplication of two inputs, as shown in Figure~\ref{fig:amulb}. 
The entire multiplication is treated as the multiplication of slices in $A$ with the slices in $B$ followed by shifting the product left by $S$ bits and accumulating the shifted result to the previous result.
There are always $N\times K$ products between elements from $f$ and $g$ that are calculated and accumulated to form the output.

\subsection{From Multiplication To Convolution}
\label{subsect:mul_conv}

To use the effective results from the product $Prod =A \times B$, we need to guarantee that the expected convolution for the different low-bitwidth inputs can be easily segmented out.
In order to segment the intermediate results, we extend the guard bits $G_b$'s definition that is introduced in Section~\ref{sec:prelim}~\cite{xilinxint4}.
It is introduced in the output to prevent the overflow of the different multiplication results.
In our proposed solution, the guard bits are not only to prevent overlaps between the effective product of two adjacent intermediate partial products but also to support the segmentation of the partial accumulations of vertically stacked partial multiplication results.
Its length varies according to the maximum number of multiplication terms $f[n]\cdot g[k]$ that are summed together. 
For a single multiplier, with $A$ and $B$ as inputs, a maximum of $min(K,N)$ terms are summed together for each output segment. 
Therefore, with our newly defined slice, to ensure the correctness of the final results, each slice should be capable of holding both the guard bits and the bits of the production from $p$-bit and $q$-bit inputs, respectively.

\begin{theorem}\label{lemma:seg}
Assuming a multiplier, with given $A$ and $B$ input multiplicands constructed from the $N$-element sequence $f$ and the $K$-element sequence $g$, where $f$ and $g$ are quantized to $p$ and $q$ bits, respectively, with the guard bits $G_b$, we can obtain $N+K-1$ segments from the product $Prod =A\times B$ which are all short partial convolutions in the form of 1-D convolution.
\end{theorem}
\begin{proof} Considering the guard bits, we can obtain:
\begin{small}
\begin{gather}\label{equ:slice1}
    S = \begin{cases}
     q+G_b, &p=1, q\geq 1\\
     p+G_b, &q=1, p\geq 1\\
     p+q+G_b, &otherwise
    \end{cases}
\end{gather}
\end{small}
\begin{equation}
\small
\label{equ:slice2}
p+(N-1)S\leq L_A
\end{equation}
\begin{equation}
\small
\label{equ:slice3}
q+(K-1)S\leq L_B
\end{equation}
\noindent Thus, we represent the A $\times$ B multiplication with K intermediate stages.
The intermediate stages are shifted left by $S$ bits for every stage, and the effective vertical accumulation of the partial products in the segments from all the stages stacked together are aligned in the $S$ bits segment, as shown in Figure~\ref{fig:amulb}. Then, the multiplication is represented as:
\begin{equation}
\small
\begin{split}
    Prod &=A \times B
    =(\sum_{n=0}^{N-1} f[n]\cdot2^{ S\cdot n })\cdot(\sum_{k=0}^{K-1} g[k]\cdot 2^{ S\cdot k })\\
    &=\sum_{m=0}^{N+K-2} ( \sum_{n+k=m}(f[n]\cdot2^{ S\cdot n }\cdot g[k]\cdot 2^{ S\cdot k }) )\\
    &=\sum_{m=0}^{N+K-2}(\sum_{n+k=m}(f[n] \cdot g[k]) \cdot 2^{S\cdot m})
\end{split}
\label{equ:multiply}
\end{equation}
Different from general multiplications, convolution consists of a sequence of multiplications and accumulations.
Referring to the form of the 1-D convolution in Equation~\ref{eq:ym}, the result of $Prod$ can be represented as:

\begin{equation}
\label{equ:prod}
\small
    Prod =\sum_{m=0}^{N+K-2} y[m]\cdot 2^{ S\cdot m}     
\end{equation}
where intermediate accumulations form a 1-D convolution of two sequences in each of the output segments, and the total number of convolution segments is $N+K-1$.
With the equations above, we can obtain the minimum value of the guard bits under the condition of a single multiplier as follows:
\begin{equation}\label{equ:gb_naive}
\small
G_b =\lceil{log_2 min(K,N)}\rceil \mid \textit{Single multiplier}
\end{equation}
to ensure that the properly accumulated partial products will not overflow.
\end{proof}

According to the above, we can use a high-bitwidth multiplier to process two integers $A$ and $B$ to form $N+K-1$ convolutions of low-bitwidth sequences.

\subsection{1-D Convolution Extension}
Now we have presented an efficient algorithm to use the multiplication unit on a hardware platform to perform the $F_{N,K}$ 1-D convolution. 
However, the size of $N$ is limited by the bitwidth of the hardware multiplier, whereas most real-world applications have much larger input sizes.
Moreover, the $F_{N,K}$ 1-D convolution is often used as a unit building block for other larger-scale convolution operations. 
Thus, we design a new algorithm to use the $F_{N,K}$ 1-D convolution to complete arbitrarily large size 1-D convolutions and any arbitrary convolutions.
As shown in Figure~\ref{fig:amulb}, the order of the elements for these intermediate production is controlled by the order of the elements packed into the slices in $A$ and $B$;
it allows us to devise different accumulation methods to provide flexibility to construct different convolutions beyond the partial convolution on a single multiplier. 

Regarding $F_{N,K}$ as a basic operation, we extend it to the $F_{X\cdot N,K}$ convolution of a longer sequence by summing up the elements in output sequences of different $F_{N,K}$ convolutions. 
\begin{theorem}\label{lemma:h_stack}
The output sequence $y=F_{X\cdot N,K}$ of a 1-D convolution between an $(X\cdot N)$-element sequence $f$ and a $K$-element filter $g$ can be represented as the sum of index-shifted output sequences $y_x=F_{N,K}(f_x,g)$, as shown in Equation \ref{eq:1dnx}. Here, $f_x=f[xN\text{:}(x\text{+}1)N\text{-}1]$ ($x \in [0, X-1]$).
\end{theorem}
\begin{proof} 
Following Equation \ref{eq:pad_f0}, we extend $f$ and $f_x$ sequences into zero extension sequences $h$ and $h_x$. Then $h$ is represented as the sum of the index-shifted sequence $h_x$:
\begin{equation}
\small
    h[n]= \sum_{x=0}^{X-1} h_x[n-xN]
\end{equation}
According to Equation \ref{eq:pad_f}, the convolution output $y$ is calculated with:
\begin{equation}
\small
\begin{split}
    y[n]&=\sum_{k=0}^{K-1}h[n-k]g[k]\\
    &=\sum_{k=0}^{K-1}(\sum_{x=0}^{X-1} h_x[n-xN-k])g[k] \\
    &=\sum_{x=0}^{X-1}( \sum_{k=0}^{K-1}h_x[n-xN-k]g[k])
\end{split}
\end{equation}
Given that $y_x[n]\text{=}\sum_{k\text{=}0}^{K\text{-}1}h_x[n\text{-}k]g[k]$, we can represent the sequence $y$ as the sum of the index-shifted $y_x$ sequences.
\begin{equation}
\label{eq:1dnx}
\small
    y[n]=\sum_{x=0}^{X-1} y_x[n-xN]
\end{equation}
\end{proof}

Equation \ref{eq:1dnx} reveals that the extended $F_{X\cdot N,K}$ 1-D convolution is computed by a shift-accumulation pattern with $F_{N,K}$ base operation results. 
Figure \ref{fig:horistack} demonstrates how the elements in different $y_x$ are summed up to the elements in $y$. Each computed $y_x$ sequence is shifted $xN$ indices and then summed up to form the element of $y$, which is marked by the red and blue squares.
In such a case, the guard bit of 
\begin{equation}
G_b = \lceil{log_2K}\rceil \mid \textit{1-D convolution}
\end{equation}
is also adjusted with additional bits to prevent the partial results from overflowing.

\begin{figure}[]
    \centering
    \includegraphics[width=0.48\textwidth]{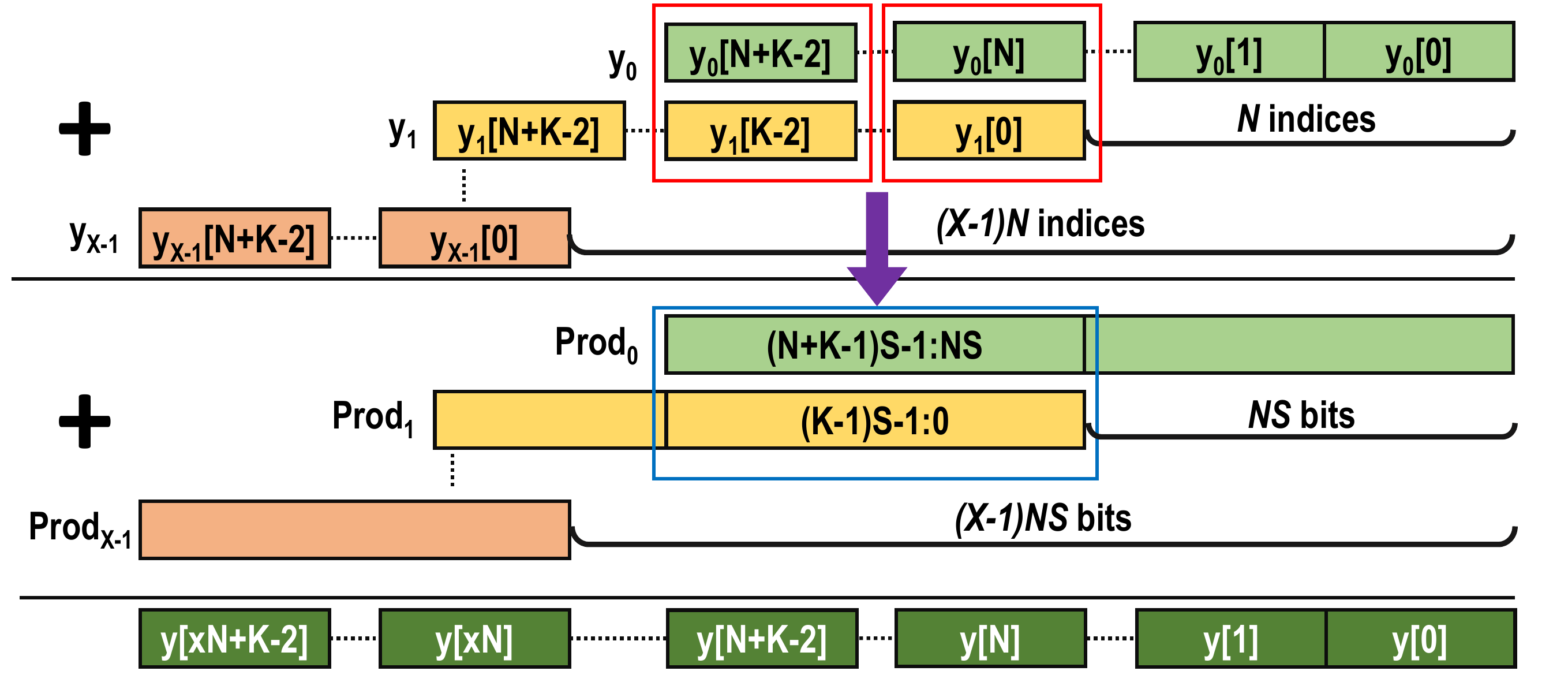}
    \caption{\textcolor{black}{Computation of $F_{XN,K}$ 1-D convolution.}}
    \label{fig:horistack}
\end{figure}

\subsection{DNN Convolution Extension}

Commonly, the convolution layer in DNN computes a feature map array $I[C_i][H_i][W_i]$ and a kernel array $W[C_o][C_i][K][K]$ for an output feature map array $O[C_o][H_o][W_o]$ (assuming $H_i\text{=}H_o\text{+}K\text{-}1$ and $W_i\text{=}W_o\text{+}K\text{-}1$) which can be represented as:
\begin{equation}
\label{eq:conv_org}
\small
\begin{split}
   O[c_o]&[h][w]\text{=} \\
   &\sum_{c_i\text{=}0}^{C_i\text{-}1}\sum_{k_h\text{=}0}^{K\text{-}1}\sum_{k_w\text{=}0}^{K\text{-}1}I[c_i][h\text{+}k_h][w\text{+}k_w]W[c_o][c_i][k_h][k_w]  
\end{split}
\end{equation}
With the inherent convolution computation pattern, we can compute a DNN convolution layer with $F_{N,K}$ 1-D convolution as the base operation, as shown in Theorem \ref{thrm:convolution}. 

\begin{theorem}
\label{thrm:convolution}
For a DNN convolution, the output feature map can be computed using the $F_{X\cdot N,K}$ 1-D convolution with the following equation:
\begin{equation}
\small
\label{eq:dnn_hikonv}
    O[c_o][h][w]=\sum_{c_i\text{=}0}^{C_i\text{-}1}\sum_{k_h\text{=}0}^{K\text{-}1} y_{c_i,c_o,h,k_h}[w\text{+}K\text{-}1]
\end{equation}
Here, the term $y_{c_i,c_o,h,k_h}$ is a 1-D convolution result with $X=\lceil{\frac{W_i}{N}}\rceil-1$:
\begin{equation}
\small
\label{eq:conv_seg}
y_{c_i,c_o,h,k_h}=F_{X\cdot N,K}(f, g) 
\end{equation}
where $f$ and $g$ are defined as:
\begin{equation}
\small
\label{eq:I_PRIME}   
\begin{split}
f[w]&=\begin{cases} I[c_i][h\text{+}k_h][w], & 0\leq h < H_i, 0\leq i < W_i\\
0, & otherwise
\end{cases}\\
g&=W[c_o][c_i][k_h][K\text{-}1\text{:}0]
\end{split}
\end{equation}
\end{theorem}
\begin{proof}
For abbreviation, we denote the sequence $y_{c_i,c_o,h,k_h}$ as $y'$. According to the definition of the 1-D convolution, the sequence $y'$ can be computed using the following equation:

\begin{equation}
\small
\label{eq:o_partial2}
\begin{split}
    y'[n]&=\sum_{k=0}^{K-1}f[n\text{-}k]g[k]\\
    &=\sum_{k=0}^{K\text{-}1}I[c_i][h\text{+}k_h][n\text{-}k]W[c_o][c_i][k_h][K\text{-}1\text{-}k]\\
    &=\sum_{k=0}^{K\text{-}1}I[c_i][h\text{+}k_h][n\text{+}k\text{-}K\text{+}1]W[c_o][c_i][k_h][k]
\end{split}
\end{equation}
Then, we have
\begin{equation}
\small
\label{eq:off}
        y_{c_i,c_o,h,k_h}[n\text{+}K\text{-}1]=\sum_{k=0}^{K\text{-}1}I[c_i][h\text{+}k_h][n\text{+}k]W[c_o][c_i][k_h][k]
\end{equation}
With Equation \ref{eq:off}, Equation \ref{eq:conv_org} could be represented as:
\begin{equation}
\small
\label{eq:1dfor3d}
\begin{split}
       O&[c_o][h][w] \\
       &\text{=}\sum_{c_i\text{=}0}^{C_i\text{-}1}\sum_{k_h\text{=}0}^{K\text{-}1}\sum_{k_w\text{=}0}^{K\text{-}1}I[c_i][h\text{+}k_h][w\text{+}k_w]W[c_o][c_i][k_h][k_w]\\ 
       &\text{=}\sum_{c_i\text{=}0}^{C_i\text{-}1}\sum_{k_h\text{=}0}^{K\text{-}1}(\sum_{k\text{=}0}^{K\text{-}1}I[c_i][h\text{+}k_h][w\text{+}k]W[c_o][c_i][k_h][k])\\
       &\text{=}\sum_{c_i\text{=}0}^{C_i\text{-}1}\sum_{k_h\text{=}0}^{K\text{-}1}y_{c_i,c_o,h,k_h}[n\text{+}K\text{-}1]
\end{split}    
\end{equation}
\end{proof}
A convolution layer in DNN has multiple input and output channels, which require accumulations of channel-wise features to form the final output. 
By grouping the $F_{N,K}$ output sequences with different $c_i$ but the same $c_o$,$h$,$k_h$, and $x$ indices and accumulating the corresponding product $Prod$,
we can perform the channel-wise accumulation of the feature maps. 
In this case, the required number of guard bits is 
\begin{equation}
G_b=\lceil{log_2(M\cdot min(K,N))}\rceil \mid \textit{DNN convolution}
\end{equation}
for the accumulation of $M$ feature maps along the input channel in a convolution.

\subsection{Maximizing Processing Throughput}
\textcolor{black}{
We can derive the total number of effective operations with respect to the multiplications and the accumulations in each $F_{N,K}(f,g)$ convolution. 
For multiplication, there are a total of $K$ times $N$ product terms summed up to form the intermediate results,
so the total number of multiplications performed in the 1-D convolution is $N\times K$. 
Meanwhile, the $N+K-1$ sets of product terms are summed up to form one element in the output sequence. 
By excluding zero products,
we calculate the total operation number in one $F_{N,K}(f,g)$ convolution in Equation \ref{eq:tot_ops}.}
\begin{equation}
\small
\label{eq:tot_ops}
\begin{split}
        \text{\# of Ops}&= \text{\# of Multiplications} + \text{\# of Accumulations}\\
        &= N\times K + (N\times K - (N+K-1) )\\
        &= N\times K + (N-1)(K-1) 
\end{split}
\end{equation}

\textcolor{black}{For each set of quantized convolution for $p$, $q$ bits input data, and a multiplier configuration with $L_A$ and $L_B$ size of inputs, we can derive the $N$, $K$ pairs that achieve the maximum number of equivalent operations performed by each multiplication within the constraints specified by Equations \ref{equ:slice2} and \ref{equ:slice3}. 
This number of operations is also equivalent to the throughput of the multiplier since they are conducted in a multiplication operation.
In order to maximize the effective number of operations for HiKonv, we formulate it as a discrete optimization problem as follows.}

Maximize: 
\begin{equation}
\small
    \# of Ops
\end{equation}

Subject to:
\begin{equation}
\small
\begin{split}
         & 0< N\\
         & 0<K\\
         p+(N-1)&(p+q+G_b)\leq L_A\\
         q+(K-1)&(p+q+G_b)\leq L_B\\
\end{split}
\end{equation}
Where in the different convolution conditions, we have different guard bits constraints:
\begin{equation}
\small
    G_b = 
    \begin{cases} 
    \lceil{log_2 min(K,N)}\rceil, & \textit{Single multiplier}; \\
    \lceil{log_2K}\rceil, & \textit{1-D convolution}; \\
    \lceil{log_2(M\cdot min(K,N))}\rceil, & \textit{DNN convolution}.
\end{cases}
\end{equation}

Then we solve this optimization problem with a straightforward search algorithm by iterating all possible pairs of $N,K$ to find the optimal configuration, as shown in Algorithm \ref{alg:optimalNK}.
After the search, the optimal configuration of $N$, $K$ with the maximum operations performed with a single multiplier is obtained.

\begin{algorithm}[t]
\caption{Optimal Throughput Search}\label{alg:optimalNK}
\footnotesize
\begin{algorithmic}[1]

\State $Max_N =(L_A-p)/(p+q)+1, Opt_K=0$
\State $Max_K =(L_B-q)/(p+q)+1, Opt_N=0$
\State $Max_{ops}=0$
\For{$k=1,k< Max_K, k\texttt{++}$}
\For{$n=1,n< Max_N, n\texttt{++}$}
\State $Cond_1= p+(n-1)(p+q+G_b)\leq L_A$
\State $Cond_2= q+(k-1)(p+q+G_b)\leq L_B$
\State $Cur_{ops}=n \times k + (n-1) \times (k-1)$

\If{$Cond_1 \text{ \& } Cond_2 \text{ \& } Cur_{ops} >  Max_{ops}$ }
\State $Opt_N=n,Opt_K=k,Max_{ops}=Cur_{ops}$
\EndIf
\EndFor
\EndFor
\State \textbf{Return} $Opt_N,Opt_K,Max_{ops}$
\end{algorithmic}
\end{algorithm}

Figure~\ref{fig:maxMul} shows two examples of multipliers with different bitwidth configurations. For a given high bitwidth processing unit, the maximum supported throughput (multiplication and addition) of a given processing unit varies with $N$ and $K$, which are determined by the values of $p$ and $q$. 
For example, when the input bitwidths of a multiplier are $27$ and $18$ bits, respectively (Figure~\ref{fig:maxMul27}), according to Equation~\ref{equ:slice1}, \ref{equ:slice2}, \ref{equ:slice3} and \ref{equ:gb_naive}, we could obtain $S=4, N=9, K=4$ when $p$ and $q$ are both 1-bit binary values. 
The maximum supported throughput of this specific multiplier is equivalent to 60 ops per cycle, which are 36 multiplications and 24 additions that are required for computing the convolution if all the computation is carried out in a conventional way following the 1-D convolution algorithm without HiKonv. Here, with HiKonv, one multiplication of high-bitwidth multiplier with our specific slicing/packing solution is equivalent to it. In addition, when $p$ and $q$ are both 4 bits, the multiplier provides 8 equivalent operations per cycle (6 multiplication and 2 addition). In Figure~\ref{fig:maxMul}, we show the configurations for $p$ and $q$ from 1-bit to 8-bit, which are the common bitwidths of low-precision quantization. The principle generally applies to all bitwidths. When the inputs for the multiplier are both 32 bits, these values are further increased to 128 ops per cycle and 13 ops per cycle for 1-bit and 4-bit $p$ and $q$, as shown in Figure~\ref{fig:maxMul32}. For the 1-D and DNN convolution layer, we only need to change the $G_b$ constraint in the searching algorithm and proceed with the straightforward search to obtain the optimal values.

\begin{figure}[]
\centering
\subfloat[A = 27 bits, B = 18 bits]{%
  \includegraphics[clip,width=0.8\columnwidth]{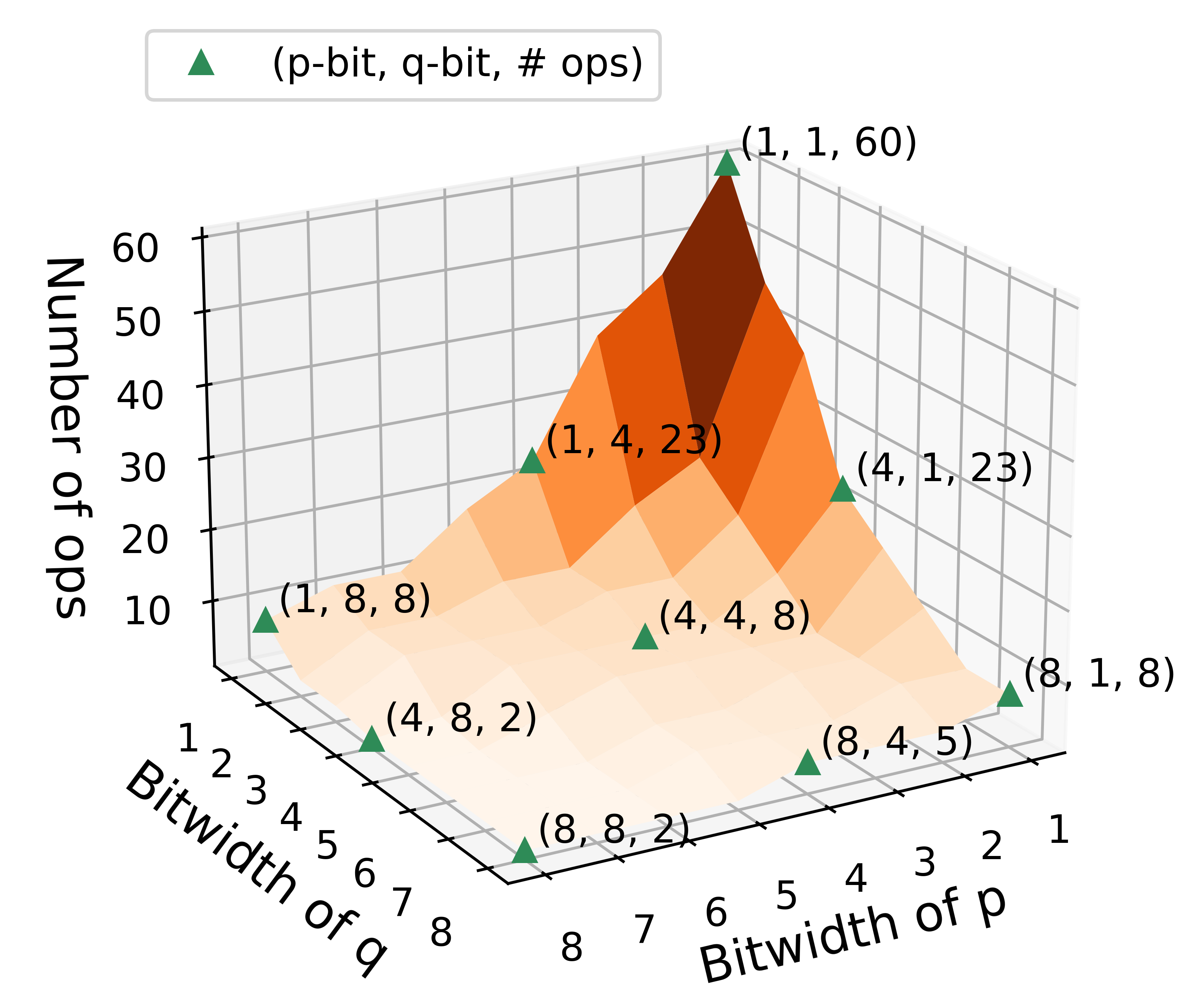}%
  \label{fig:maxMul27}
}

\subfloat[A = 32 bits, B = 32 bits]{%
  \includegraphics[clip,width=0.8\columnwidth]{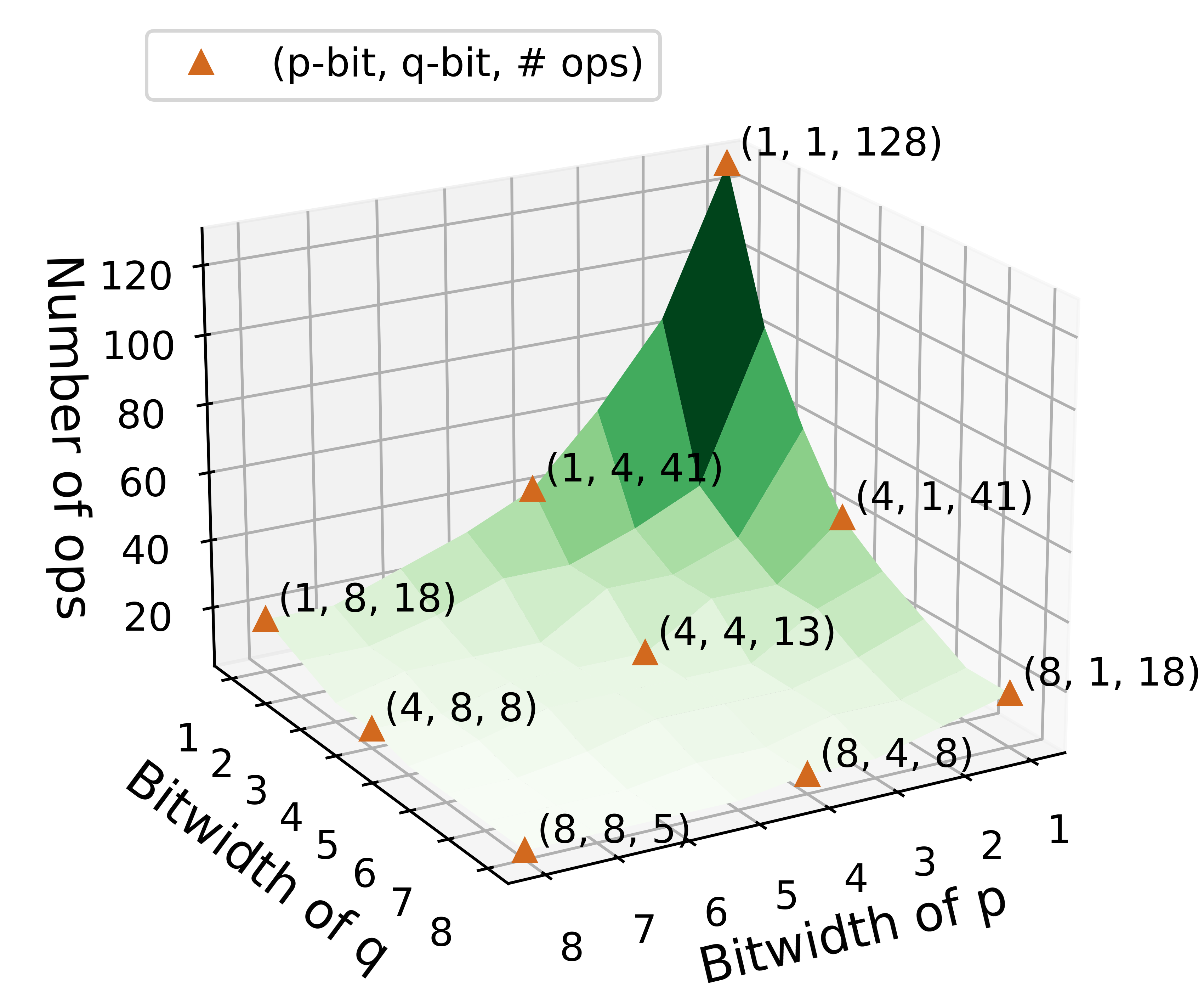}%
  \label{fig:maxMul32}
}
\caption{Throughput of processing units with different bitwidth settings.}
\label{fig:maxMul}
\end{figure}

\subsection{\textcolor{black}{Sign Extension for Sign Integers}}

\textcolor{black}{
HiKonv differs from conventional multi-channel multiplication methods~\cite{xilinxint4, xilinxint8} with a more comprehensive bit-wise management of the input data and more efficient use of the outputs.
To implement HiKonv on the existing computing platforms, including GPP and FPGA,
the overall processes are abstracted as \textit{input packing}, \textit{large bitwidth multiplication} and \textit{output split}.
Regarding Equation~\ref{eq:inpack},
packing the input is equivalent to extending low-bitwidth data to higher bitwidth and then performing accumulation of multiple extended data.
However, the two's complementary representation in modern computing systems performs sign extension when low-bitwidth data is extended to a larger bitwidth length.
Therefore, in the form of binary representation, the low-bitwidth input that has been packed into the left segment is a combination of the low-bitwidth data itself with the sign extension of the data in the segment to the right of it.}

\textcolor{black}{
Same as the input packing,
extraction of $y[m]$ from the $Prod$ can not be simply processed with bit-wise segmentation
because the previous segment's sign is continuously added to the current segment.
More specifically, for $y[m+1]$, if $y[m]$ is a negative value, its sign extension stands for $-1$, which in turn passes the sign extension to the segment on its left.
The sign of the segments must be considered to extract the correct value of $y[m]$.
}
\section{Platform Oriented Implementations}\label{sec:impl}

Compared to the multiplication operation,
the above packing and split operations have different architectural preferences due to the unique bit-wise operations, such as adopting shift registers or using unique hardware logic for efficient processing.
In this section, we provide detailed implementations of HiKonv regarding the architectural characteristics of the general-purpose processors and reconfigurable platforms.

\subsection{HiKonv on General Purpose Processors}

We provide two versions of the HiKonv implementation on general-purpose processors, including both Intel X86\_64 CPUs and ARM processors. 
Modern CPUs are usually 64-bit processors that are equipped with 32-bit integer multipliers. With such an architecture, a 32-bit multiplication is performed with 32-bit registers as operands while 64-bit is constructed using two 32-bit registers. 
The 64-bit product is stored in two 32-bit registers: the upper half in one and the lower half in the other. 
A 64-bit multiplication is done in the same way, except that the registers are 64-bit registers. 
Thus, without loss of generality, we use 32-bit multiplication for our HiKonv implementations.

\subsubsection{Basic HiKonv Operations}

\begin{algorithm}[t]
\caption{1-D HiKonv with single multiplier on GPP}\label{alg:cpuimpl}
\footnotesize
\begin{algorithmic}[1]
\For{$n=0, n<N, n\texttt{++}$} \Comment{Pre-packing}
\State $A = A \texttt{ + } f[n] \texttt{ << } (s*n) $
\EndFor
\State
\State $prev = A*B$ \Comment{Prologue multiplication}
\State
\For{$c=N,c<M \texttt{ + } N, c\texttt{++}$} 

\setlength{\parindent}{-.65em}
\colorbox[RGB]{239,240,241}{
\begin{minipage}{0.98\linewidth}
\For{$n=0,n<N, n\texttt{++}$}\Comment{\textbf{Input data packing}}
\State $A = A \texttt{ + } f[c\texttt{+}n] \texttt{ << } (s*n) $
\EndFor
\end{minipage}
}
\State
\State $this = A*B$\Comment{Multiplication}

\setlength{\parindent}{-.65em}
\colorbox[RGB]{255,189,200}{
\begin{minipage}{0.98\linewidth}
\State $y = prev \texttt{ << } (s*(N\texttt{-}1)) \texttt{ + } this \texttt{ << } s $
\LineComment{\textbf{Intermediate product shifting}}
\end{minipage}}
\State

\setlength{\parindent}{-.65em}
\colorbox[RGB]{200,228,250}{
\begin{minipage}{0.98\linewidth}
\For{$n=0,n<N, n\texttt{++}$}\Comment{\textbf{Output data splitting}}
\State $output[c \texttt{-} n] = (y \texttt{ >> } (s*(N\texttt{-}n\texttt{-}1))) \texttt{ }\&\texttt{ } (2^s \texttt{-} 1) $
\State $// \texttt{ } carried\_sign = y \texttt{ >> } (s*(N\texttt{-}n\texttt{-}1) - 1)$
\State $// \texttt{ } output[c \texttt{-} n] = output[c \texttt{-} n] + carried\_sign$
\EndFor \Comment{Signed HiKonv: Check for the sign bit}
\end{minipage}}

\State \noindent
\hspace{-0.5em}\EndFor
\State \hspace{-1.5em}\textbf{Return} $output$
\end{algorithmic}
\end{algorithm}

HiKonv supports operations between positive and negative integer inputs.
Meanwhile, most activations and weights in many modern neural networks can be trained to only positive numbers. It is beneficial to utilize an unsigned version of HiKonv as the implementation is simpler than the signed HiKonv, which requires sign checking in both input packing and output split (detailed in the following subsections).
Due to the different implementation requirements, we provide two sets of HiKonv implementations, which are the HiKonv implementation that only supports unsigned input and the HiKonv implementation that supports signed inputs, on general-purpose processors. 
The unsigned HiKonv implementation only operates with positive integer inputs, and the signed HiKonv implementation supports both positive and negative inputs.
The implementations of the three main processes are shown in Algorithm~\ref{alg:cpuimpl}.
Except the major \textit{input packing}, \textit{multiplication} and \textit{output split} operations, there are additional operations named as \textit{Pre-packing} and \textit{Alignment} due to the efficient shifting and store-and-accumulate operations on the ALU in GPPs.
We explain each step in the following.
\begin{enumerate}
    \item \textit{Input data packing}
    According to Equation~\ref{equ:multiply}, we can construct the large bitwidth inputs with shift-add operations. For an input vector with $N$ segments, we need $N-1$ shift-add operations to pack the input data properly. To pack an input, $f[n]$ needs to be shifted by $s\times n$ bits to the left and then accumulates to the input bit-vector. Packing the other input to the multiplier follows the 
   same procedure. In the context of computing convolutional neural networks, the features need to be packed during the runtime, whereas the weights can be pre-packed, enabling further optimization opportunities. This step is the same for both the signed and the unsigned implementations of HiKonv.
    \item \textit{Intermediate product shifting}
    We adopt the horizontal stacking strategy discussed in Section~\ref{sec:formula}. As seen in Figure~\ref{fig:amulb}, the output of one block depends on the multiplication results from both the present and the previous iterations. The multiplication operation is inherently multiply-accumulate between the two outputs; however, the two results need to be aligned and accumulated to produce the output correctly in an output segment. Considering the limited length of the shift registers in GPPs, we shift the partial result from the previous iteration to the right and the partial result from the current iteration to the left by the appropriate number of bits, respectively. Once they are properly aligned, we add them together to construct the complete result for this iteration and extract the bits for the output.
    \item \textit{Output data splitting}
    This step differs between the signed and unsigned implementations, as shown in Algorithm~\ref{alg:cpuimpl}. Due to the sign extension in GPP platforms, dealing with signed values requires checking the sign bits in different segments, where we take an additional step to check the sign of the preceding output. As seen in Equation~\ref{equ:output_split}, the most effective way is to shift the output and check the sign bit immediately preceding the current output segment. If the sign equals to one, indicating that the segment to the right is a negative number, we need to compensate and adjust the current output. In the unsigned version of HiKonv, since all numbers are positive, there is no need to perform the sign bit checking.
\end{enumerate}

\subsubsection{1-D Convolutions with HiKonv}
For 1-D convolution, the baseline implementation has nested loops of two levels. The outer loop scans through the input sequence, whereas the inner loop scans through the kernel sequence.
In the implementation of 1-D convolution with HiKonv, we only need one loop to compute the convolution, as demonstrated by Algorithm~\ref{alg:cpuimpl}. We need a prologue to perform the first multiplication operation and set up the convolution iterations. In each of the iterations, it performs the basic three operations in HiKonv, including input data packing, intermediate product shifting, and output data splitting. There is only one multiplication in each loop; the number of bits to shift and bit-masks can be pre-calculated and pre-defined offline based on Algorithm~\ref{alg:optimalNK}.

\textcolor{black}{To better illustrate the 1-D convolution with HiKonv, we provide the following numerical example in Figure~\ref{fig:nume_example}. We demonstrate a $F_{3,2}$ convolution where we convolve the three-element array $[11, 9, 7]$ with a two-element kernel $[3, 2]$. The result of the convolution should be $[33, 49, 39, 14]$. On the left of Figure~\ref{fig:nume_example}, we show how we pack the inputs into two 32-bit bitstrings and perform the multiplication with binary values. On the right, we show the ideal process in decimals. It is easy to compare and see how the convolution can be converted into one multiplication of which the product contains the results of the convolution output. In practice, the binary multiplicands are 11543559 and 3074 in decimal. Their product is 35484900366, which is 100001000011000100001001110000001110 in binary as shown on the left of Figure~\ref{fig:nume_example}.}

\begin{figure}
    \centering
    \includegraphics[width=0.48\textwidth]{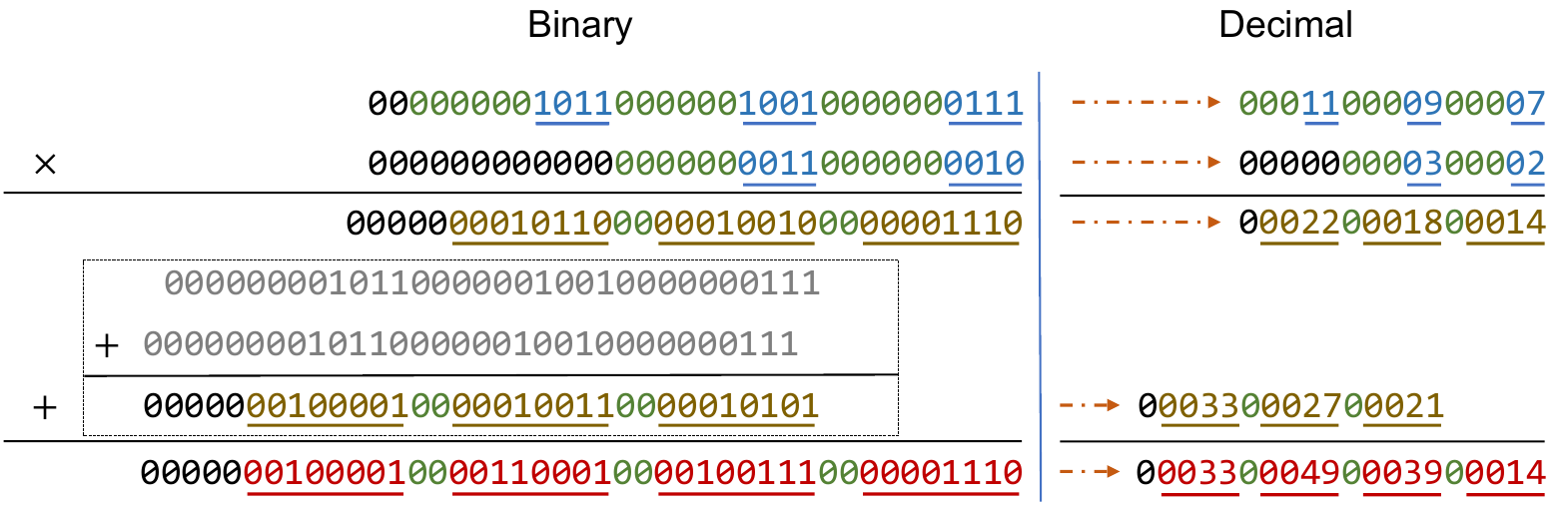}
    \caption{A numerical example of a $F_{3,2}$ 1-D convolution.}
    \label{fig:nume_example}
\end{figure}

\subsubsection{2-D Convolutions with HiKonv}
We implement the DNN layer by embedding the 1-D convolution in the six-level nested loops that scan through the input channel, output channel, output height, output width, kernel height, and kernel width according to Theorem~\ref{thrm:convolution}. The order of the nested loop does not have an impact on the functionality. In order to fit the 1-D convolution algorithm into the context of 2-D convolution, we arrange the order of the loops in such a way that we first compute the partial results of each row and then accumulate them in the channel dimension. HiKonv supports even higher dimensional convolution with a similar approach.

\subsection{HiKonv on Reconfigurable Hardware}
Reconfigurable hardware provides a finer granularity control of the data path down to a single bit.
We take advantage of the flexibility in both input data packing and output splitting with a small number of additional logic resources to further improve the effectiveness of the required operations.
Particularly, we target the Xilinx FPGA platforms, which provide DSP48E2 resource that contains a $27\times18$ bits multiplier~\cite{xilinxint4}.

\subsubsection{HiKonv for single DSP}
To deploy HiKonv on a single FPGA DSP, we first explore the optimal $N$, $K$, and $S$ values with Algorithm~\ref{alg:optimalNK} based on the target bitwidth of the inputs, and configure the hardware unit with these parameters.

\begin{figure}[]
    \centering
    \includegraphics[width=0.48\textwidth]{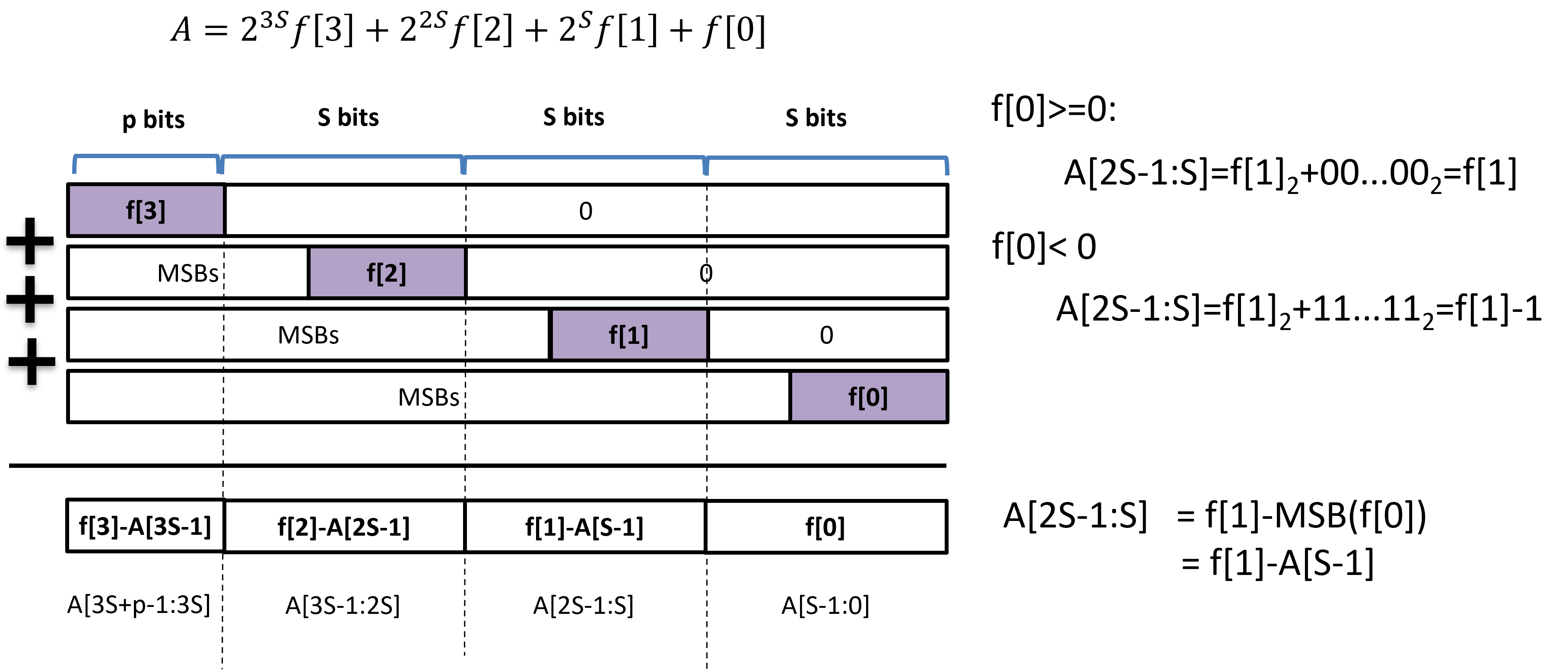}
    \caption{Input packing for signed integer $f$ sequence.}
    \label{fig:signpack}
\end{figure}

\begin{enumerate}
\item \textit{Input data packing} 
Since the packing of the elements is equivalent to performing a long bitwidth shift add operation,
if $f$ and $g$ are all unsigned integers, we can construct $A$ and $B$ with bit-wise assignments to the corresponding slice with a simple zero extension:
\begin{equation}
\small
\begin{split}
   &A[S(n\text{+}1)\text{-}1\text{:}Sn]= f[n]\\
    &B[S(k\text{+}1)\text{-}1\text{:}Sk]= g[k]
\end{split}
\end{equation}

When $f$ and $g$ contain signed integers, 
instead of performing an extension-shift-addition operation, we simplify the logic to check the sign bit of the low bitwidth data before packing it into the slices.
As shown in Figure~\ref{fig:signpack},
taking the second slice as an instance, in two’s complement expression, if $f[0]$ is positive, the MSB is 0, and the sign extension part is all zero.
On the other hand, if $f[0]$ is negative, the sign extension part is all $1$ in the binary expression and represents -1 in the two's complement representation.
To take advantage of the bit flexibility of FPGA devices, we still perform a bit-wise assignment of the input data to the slice but with an online checking of the sign value of the previous slice.
Equation \ref{eq:signpackslice} shows the packing formula for the signed integers from $f$ and $g$ into the $A$ and $B$ multiplicands.
\begin{equation}\label{eq:signpackslice}
\small
\begin{split}
      A[S(n\text{+}1)\text{-}1\text{:}Sn]= \begin{cases}
    f[0] &, n=0\\
    f[n]\text{-}A[Sn\text{-}1]&, n>0
    \end{cases}\\
    B[S(k\text{+}1)\text{-}1\text{:}Sk]= \begin{cases}
    g[0] &, k=0\\
    g[k]\text{-}B[Sk\text{-}1]&, k>0
    \end{cases}
\end{split}
\end{equation}

In such a condition, when packing $f[1]$ into the second slice, we can decrement 1 from $f[1]$ to obtain the value in that slice.
To simplify the logic, there only requires a simple 1-bit decrementer before the concatenation of the changed values to form the entire multiplicand; whether enabling this 1-bit decrementer is based on the sign bit from the previous slice, as shown as the \textit{Packing Decrementers} in Figure~\ref{fig:convmicroarch}.
Note here, due to the pre-adder in the Xilinx DSP48E2 module, resource for one of the \textit{Packing Decrementers} can be saved with this pre-adder, but it is omitted here to simplify the presentation.
The packing process works recursively for all the slices.
This approach saves the resource from using large bitwidth shift registers and adders and, in return, only costs a small decrementer and single-bit Boolean logic.

\item \textit{Large bitwidth production}
The production is done with the adoption of the multiplier in the FPGA DSP, noted as \textit{DSP multiplier} in Figure~\ref{fig:convmicroarch}; the packed weight sequence and feature sequence are directly provided to the 27-bit and 18-bit input data ports of the multiplier. The output is obtained in a single clock cycle with a pipeline depth of 1. In this way, with a single multiplication, multiple low-bitwidth multiplications and additions are performed.

\item \textit{Output split}
Output split is a reversed process of packing the data into slices.
The final output of the production is formed by multiple $y[m]$ values. Each of the current S-bit segmentation carries on the impact of the sign bit from the S-bit segmentation to the right of it, except for the very right segmentation, as shown in Equation~\ref{equ:output_split}.
Due to the two's complement representation of the values in modern computing systems, the sign extension mechanism extends the sign value of current $y[m]$ to the MSB of the final output, so the value of each $y[m]$ is a function of the current value in the current segmentation with the sign value in the previous segmentation, as shown in the Equation~\ref{equ:output_split}. Similarly, the implementation of Equation~\ref{equ:output_split} adopts a simple incrementer as the input packing instead of using a shift-and-add module that requires more logic resources.
\begin{equation}
\small
\begin{split}
    y[m]= \begin{cases}
    Prod[S\text{-}1\text{:}0],& m=0\\
    Prod[S(m\text{+}1)\text{-}1\text{:}Sm]\text{+}Prod[Sm\text{-}1],& m>0
    \end{cases}
\end{split}
\label{equ:output_split}
\end{equation}
\end{enumerate}

\subsubsection{Single convolver architecture}

\begin{figure}[ht]
    \centering
    \includegraphics[width=0.4\textwidth]{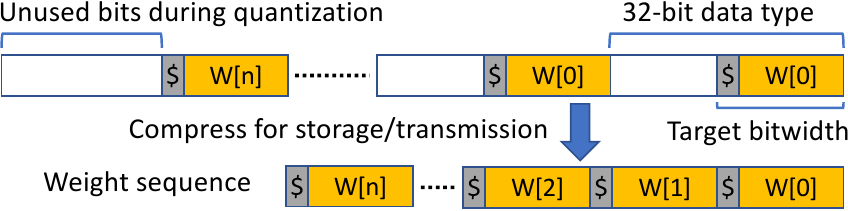}
    \caption{Weight compression.}
    \label{fig:weightcompression}
\end{figure}

For quantized convolution in DNNs, weight data can be compressed prior to inference by discarding the unused bits during quantization~\cite{vecq, ul2q, tdla}, as shown in Figure~\ref{fig:weightcompression}.
The compressed weights reduce the memory access overhead by a fraction of the reduced bits and also simplify the data transfer logic to save hardware resources. For example, compressing the weights quantized to 4 bits reduces memory to $1/8$ compared to the original 32 bits data type for storage.
The same compression can also be applied to the feature data during runtime.
Both weight and feature data can be extracted at runtime from the compact weight sequence and feature sequence and packed into the slices of the multiplier's inputs without affecting throughput.
With the integration of the input data pack and the output split, the microarchitecture of a new convolver that is mainly built with FPGA DSPs is shown in Figure~\ref{fig:convmicroarch}.

A certain number of elements of the compressed weight sequence and the feature sequence are first buffered in \textit{Input Registers}, and then passed to \textit{Packing Decrementer}, which is constructed with multiple decrementers, as discussed in the previous section. The two inputs to the multiplier are then packed as multiplicands and passed to the multiplier of the DSP in the FPGA chip. The number of decrementers is based on the previously explored $N$ and $K$ with Algorithm~\ref{alg:optimalNK}. The decrementers provide the packed values in a single clock cycle, and $Prod$ is obtained after one clock cycle from the DSP and passed to the incrementer logic to obtain the final output.

With this architecture, we obtain $N + K -1$ partial convolutions from $N \times K$ segments of intermediate results from a single multiplier to process signed input data. While for unsigned data processing, the \textit{Packing Decrementers} and \textit{Split Incrementers} can all be reduced, which in return occupies less logic resource.

\begin{figure}[ht]
    \centering
    \includegraphics[width=0.48\textwidth]{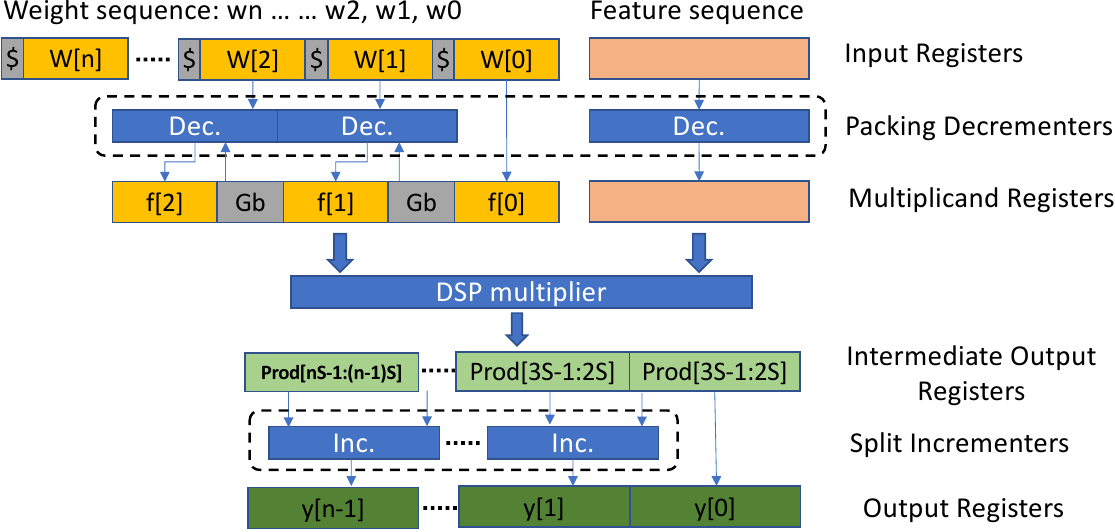}
    \caption{Micro-architecture of a single DSP convolver.}
    \label{fig:convmicroarch}
\end{figure}

\subsubsection{1-D \& DNN layer Convolution Architecture}\label{sec:1d2dconv}
The HiKonv 1-D convolution on FPGA is based on Equation~\ref{eq:1dnx}. Different from the single DSP convolver, the guard bits here need to be adjusted to $G_b = \lceil{log_2K}\rceil$ so that the partial additions do not overflow. 
After the adjustment of the guard bits, a single convolver is recursively used to process the input and weight sequences. It requires an additional register to temporarily store the output result of current multiplication to continuously construct new convolution outputs by shifting the stored data and adding it with the new output from the DSP multiplier, as shown in Figure~\ref{fig:1dconv}.

\begin{figure}[]
    \centering
    \includegraphics[width=0.4\textwidth]{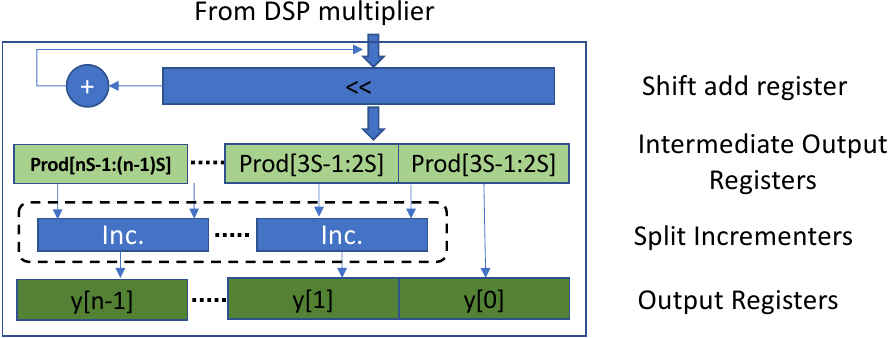}
    \caption{FPGA HiKonv 1-D convolution output processing.}
    \label{fig:1dconv}
\end{figure}

However, for a convolutional layer in a DNN, a single output involves data from multiple input features even after the input channels are tiled; based on the Equation~\ref{eq:1dfor3d}, we first change the guard bits to $G_b=\lceil{log_2(M\cdot min(K,N))}\rceil$, so the $S$-bit slice is enough to hold the accumulation results from M feature maps. This also enables us to perform the addition operation between the different input channels before \textit{Split Incrementers} to further reduce the number of 1-bit incrementers to save resources, as shown in Figure~\ref{fig:2dconv}.

\begin{figure}[t]
    \centering
    \includegraphics[width=0.48\textwidth]{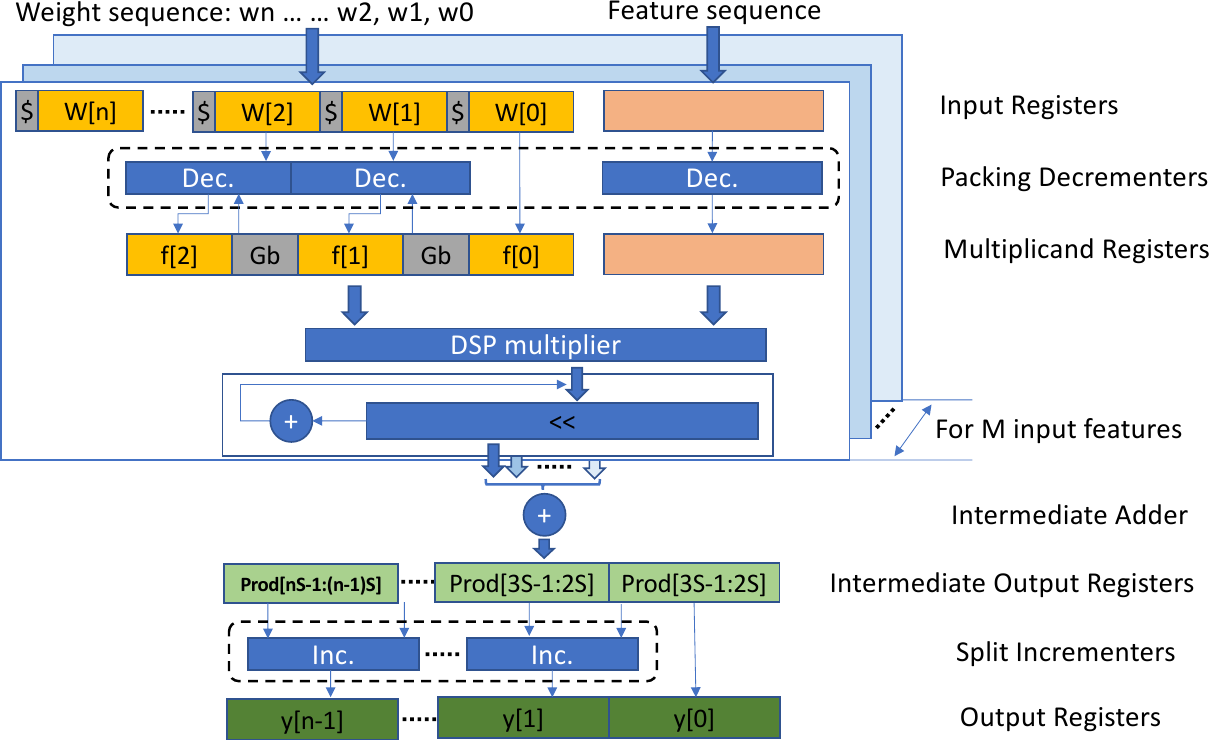}
    \caption{FPGA HiKonv 2-D convolution.}
    \label{fig:2dconv}
\end{figure}
\section{Evaluations}\label{sec:eva}

HiKonv is a general technique that can be adopted for both the GPP and reconfigurable hardware platforms. We demonstrate its efficacy on both platforms with the proposed implementations in Section~\ref{sec:impl}.

\subsection{HiKonv on GPPs}
To demonstrate the generality of our proposed HiKonv solution,
we test the implementations for GPP on the Intel desktop CPU platform, Intel mobile CPU platform, and ARM platforms.
Specifically, the implementations are evaluated on the Intel Core i7-10700K CPU, the i7-10710U CPU, and the Raspberry PI 3B+ platform with an ARM Cortex A53 processor, respectively.

\subsubsection{Convolution layer evaluation}

\begin{figure}
        \begin{subfigure}{0.45\textwidth}
        \centering
        \includegraphics[width=\linewidth]{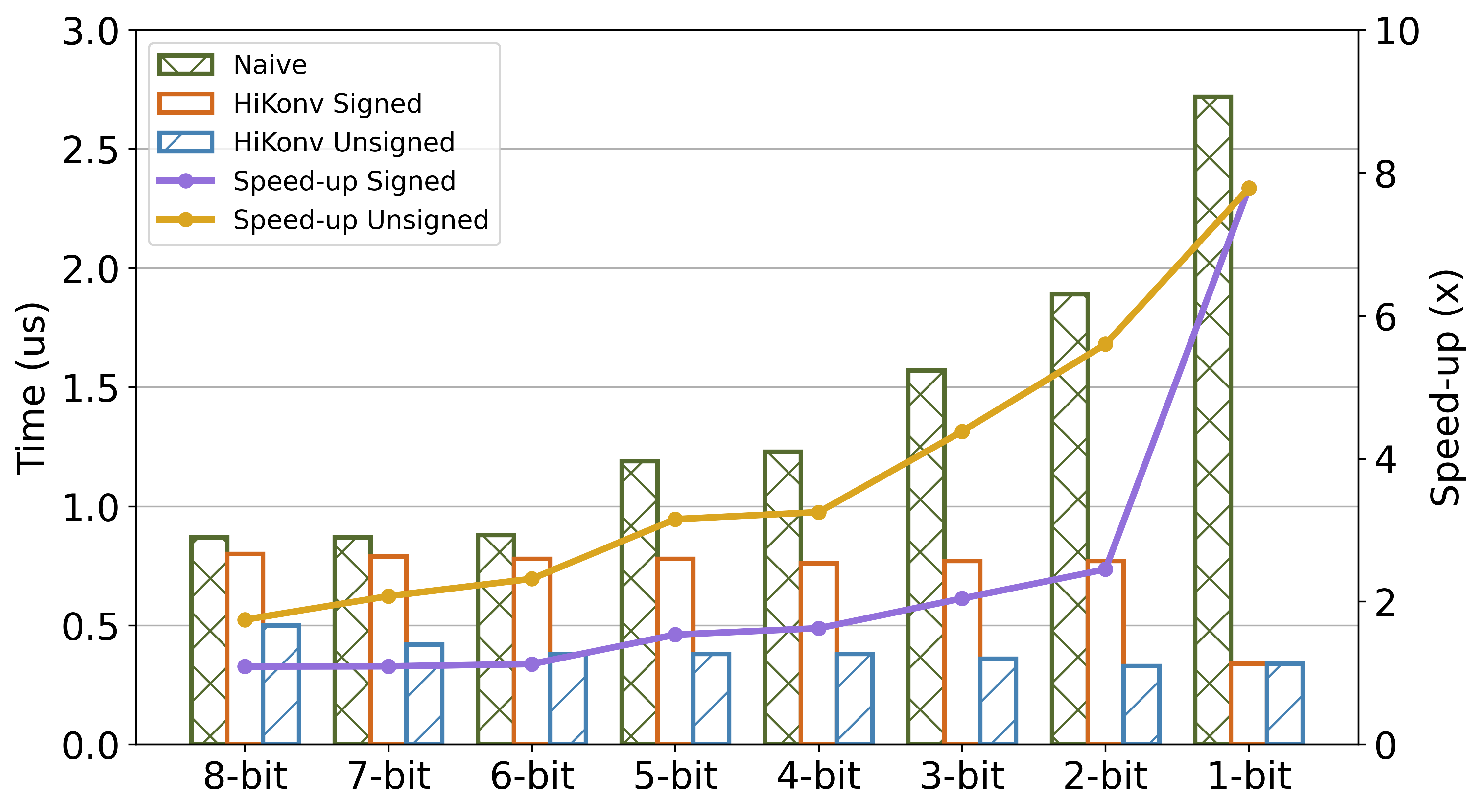}
        \caption{Speedup on X86\_64.}
        \label{fig:all_bits}
    \end{subfigure}
    \begin{subfigure}{0.45\textwidth}
        \centering
        \includegraphics[width=\linewidth]{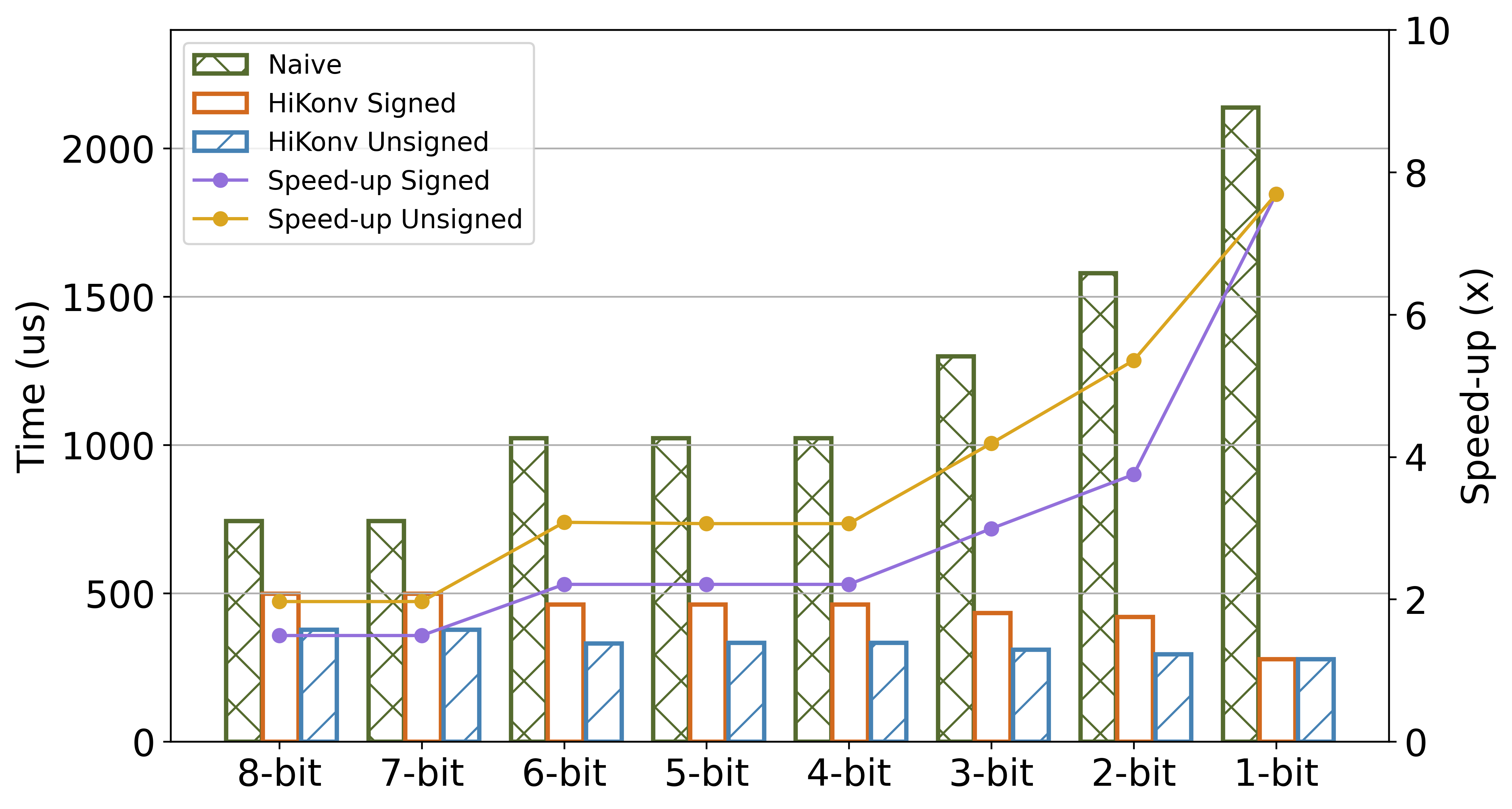}
        \caption{Speedup on ARM.}
        \label{fig:allbitarm}
    \end{subfigure}
    \caption{Speedup for different bitwidth.}
    \vspace{-4mm}
    \label{fig:arm_eval}
\end{figure}

We first measure the performance of the 1-D convolution with low-precision bitwidths from 1 to 8 bits. For quantified analyses, we randomly generate feature and kernel vectors.
Since modern CPUs are equipped with 32-bit multipliers, without loss of generality, we use $A=B=32$ bits as the multiplication bitwidth.
Assuming $p = q$, we calculate the corresponding $N$, $K$, and $G_b$ and pack the quantized values into 32 bits accordingly. Figure~\ref{fig:all_bits} and Figure~\ref{fig:allbitarm} show the result of the 1-D convolution for both X86\_64 CPUs and ARM processor.
It is clear that when the bitwidth of the processed data is reduced, the performance of our HiKonv solution increases because of the increased number of operations. When the bitwidth is 1, the HiKonv solution provides a $7.8\times$ performance improvement on the X86\_64 platform for both signed and unsigned data inputs.
While the improvement goes down to $1.2\times$ and $1.8\times$ when the bitwidth is increased to 8 bits.
The same performance improvements are maintained on the ARM processor, which are $7.6\times$ for 1-bit inputs for both signed and unsigned data inputs and $1.4\times$ and $1.9\times$ when the inputs are 8 bits. These performance improvements demonstrate the generality of our HiKonv solution.

We then evaluate the performance of the 1-D convolution together with the 2-D DNN layer. For the 2-D layer, we pick the final layer of UltraNet~\cite{UltraNet2020}, which is the champion model for the DAC-SDC contest 2020, and randomly generate feature and kernel sequences.
Same as the previous experiment, we use $A=B=32$ bits as the multiplication bitwidth and pack $p=q=4$-bit unsigned values in each of the operands. According to Theorem~\ref{lemma:h_stack}, we obtain $N = 3, K = 3, G_b = 2$ and $S = 10$ bits. Figures \ref{fig:cpu_eval_1d} and \ref{fig:cpu_eval_2d} show the 1-D and 2-D convolution latency results on the Intel platforms, respectively. Both are compared to the baseline implementation with nested loops without our HiKonv solution, noted as ``Naive" in the figure. 
The implementations to support signed and unsigned values are slightly different from each other, and we evaluate both of them, noted as Signed and Unsigned in the figure.
CPU hardware lacks bitwise management capability, and dealing with signed values can cause overhead from intricate bit operations.
HiKonv-enabled convolution implementations with signed integer support perform at least $2.26\times$ faster than the naive ones and reach $3.21\times$ faster for only unsigned support. 
For 2-D convolutions, HiKonv performs $2.74\times$ for signed data input and achieves $3.19\times$ speed-up for unsigned implementation compared to naive implementation.
The implementations that only support unsigned integer input are always faster because there is no sign bit checking during output data split.

The same performance improvement has been observed on the ARM processor, as shown in Figure~\ref{fig:arm_eval_1d} and Figure~\ref{fig:arm_eval_2d}. Compared to naive implementation on an ARM processor, the speedups for signed and unsigned data input reach $2.21\times$ and $3.06\times$ for 1-D convolution and $2.75\times$ and $2.98\times$ for 2-D convolution. Although all results are magnitude slower than the implementations on Intel X86\_64 CPUs, the speedups are only slightly reduced due to the constrained cache size and memory access capability on the ARM processors.

\begin{figure}
\centering
    \begin{subfigure}{0.22\textwidth}
        \centering
        \includegraphics[width=\linewidth]{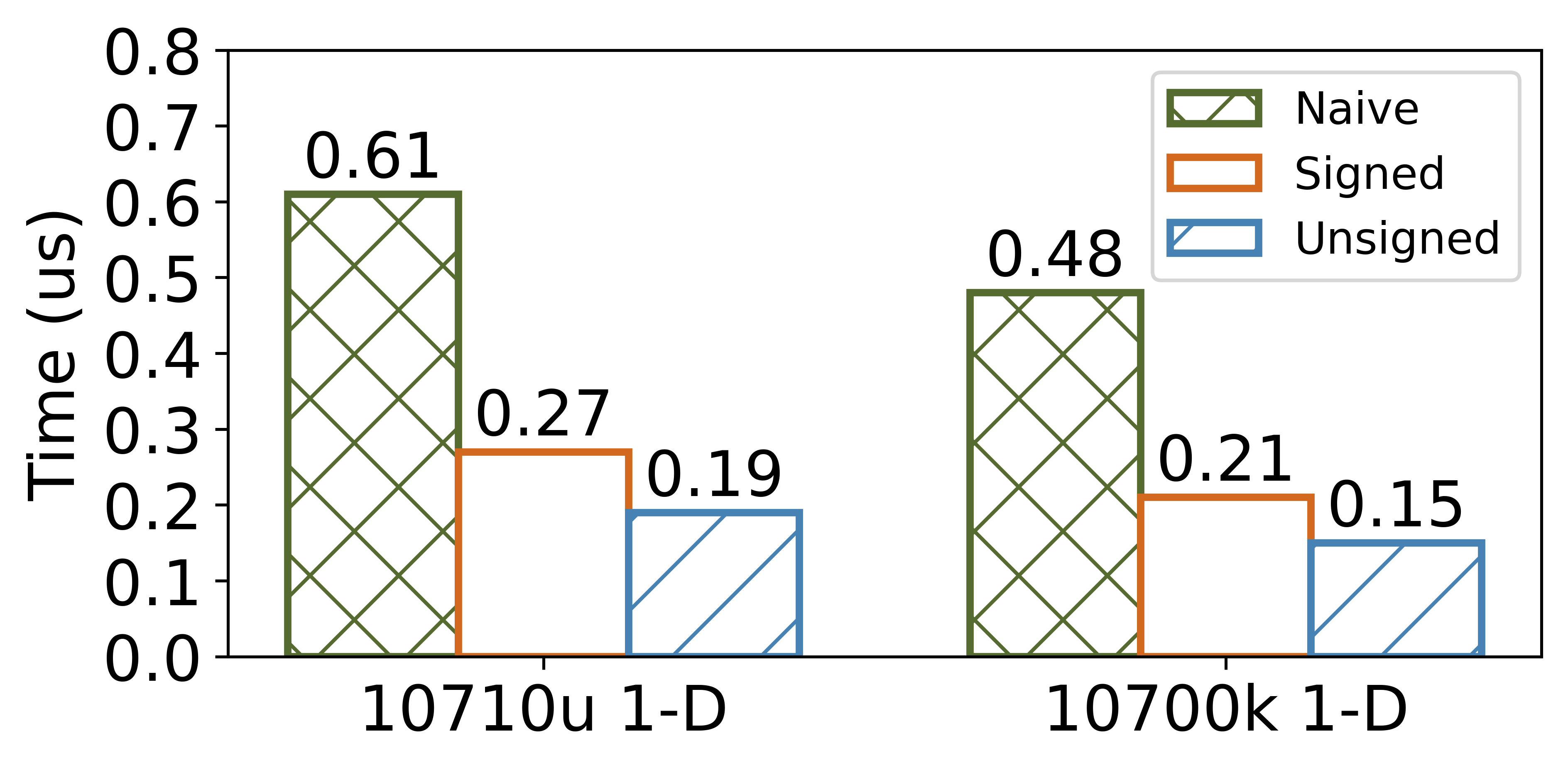}
        \caption{1-D Convolution.}
        \label{fig:cpu_eval_1d}
    \end{subfigure}
    \hfill
    \begin{subfigure}{0.22\textwidth}
        \centering
        \includegraphics[width=\linewidth]{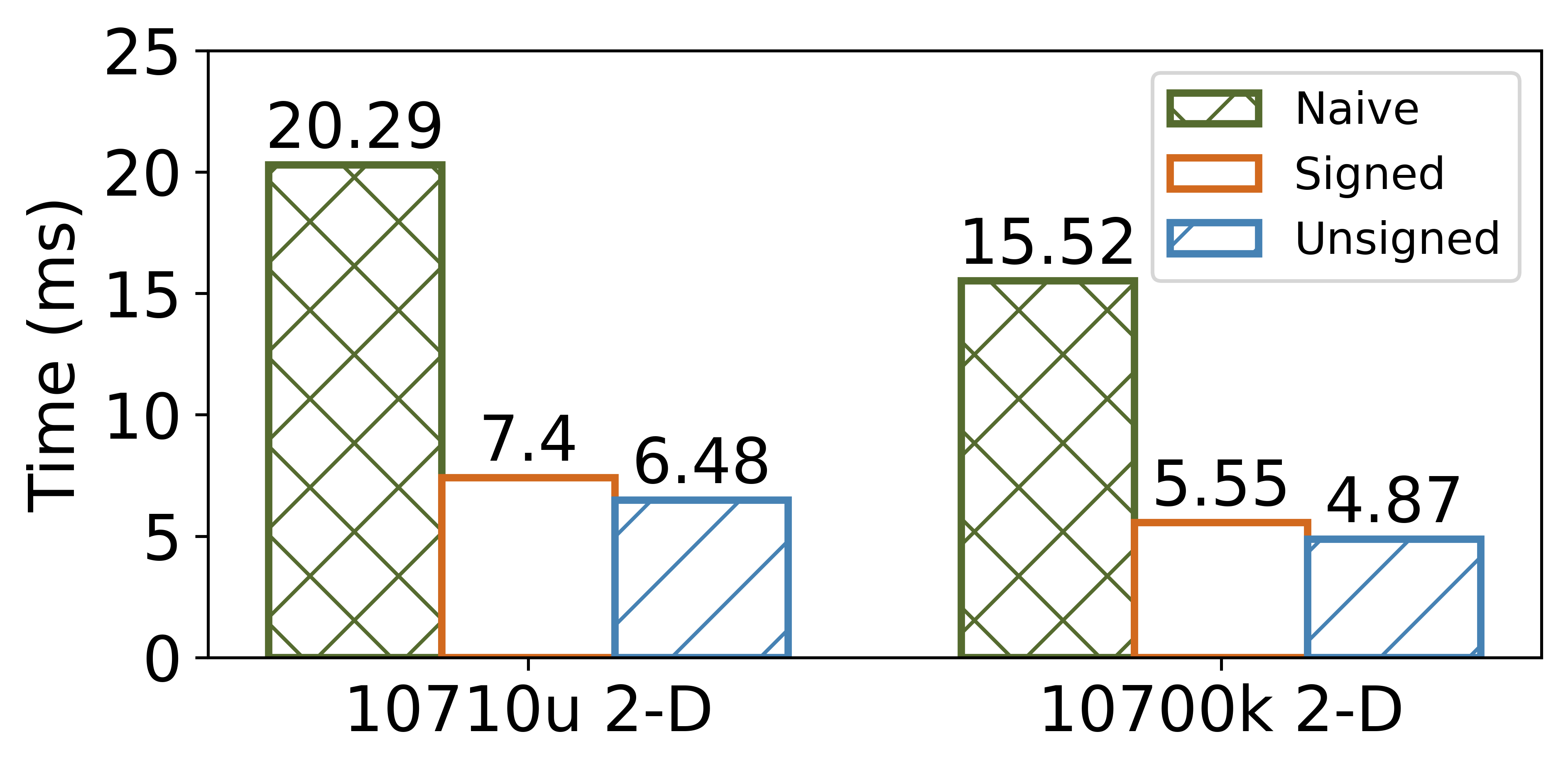}
        \caption{2-D Convolution.}
        \label{fig:cpu_eval_2d}
    \end{subfigure}
    \label{fig:intel4bit}
    \caption{Evaluation with 4-bit layers on X86\_64 CPUs.}
    \vspace{-4mm}
\end{figure}

\begin{figure}
\centering
    \begin{subfigure}{0.2\textwidth}
        \centering
        \includegraphics[width=\linewidth]{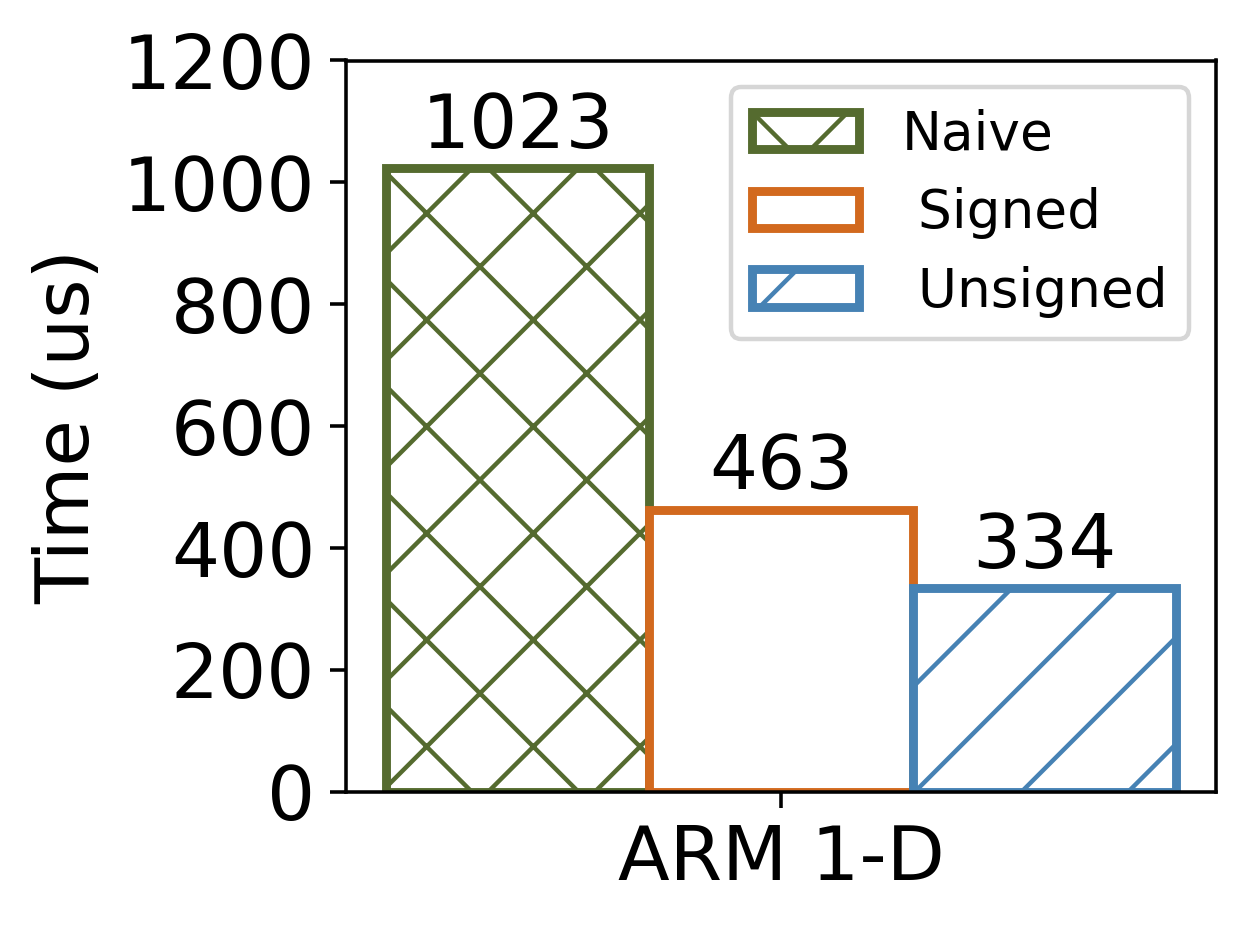}
        \caption{1-D Convolution.}
        \label{fig:arm_eval_1d}
    \end{subfigure}\hspace{1em}%
    \begin{subfigure}{0.2\textwidth}
        \centering
        \includegraphics[width=\linewidth]{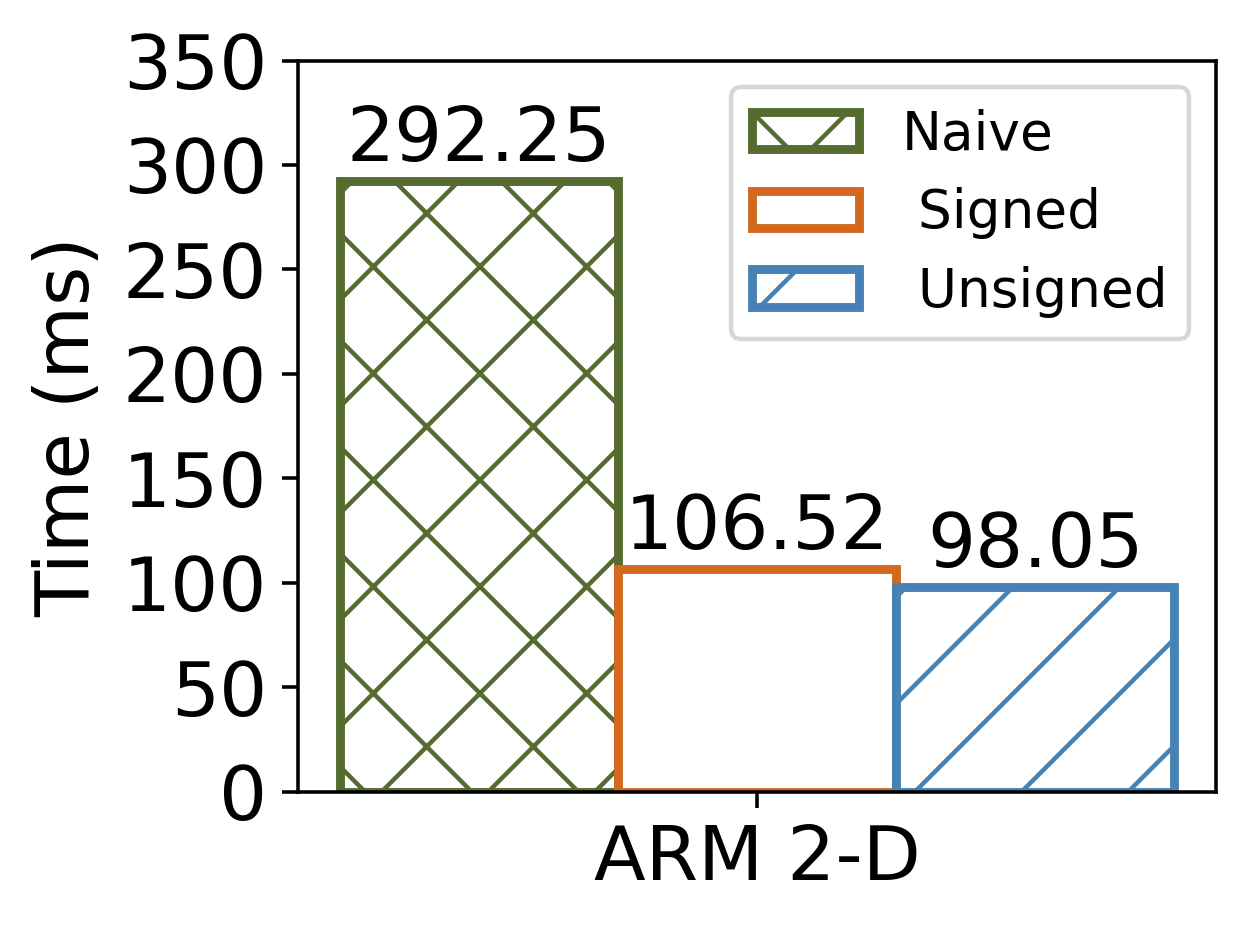}
        \caption{2-D Convolution.}
        \label{fig:arm_eval_2d}
    \end{subfigure}
    \caption{Evaluation with 4-bit layers on ARM.}
    \vspace{-4mm}
    \label{fig:cpu_eval}
\end{figure}

\subsubsection{Complete model evaluation}\label{sec:fullmodelcpu}

\begin{figure}
        \begin{subfigure}{0.45\textwidth}
        \centering
        \includegraphics[width=\linewidth]{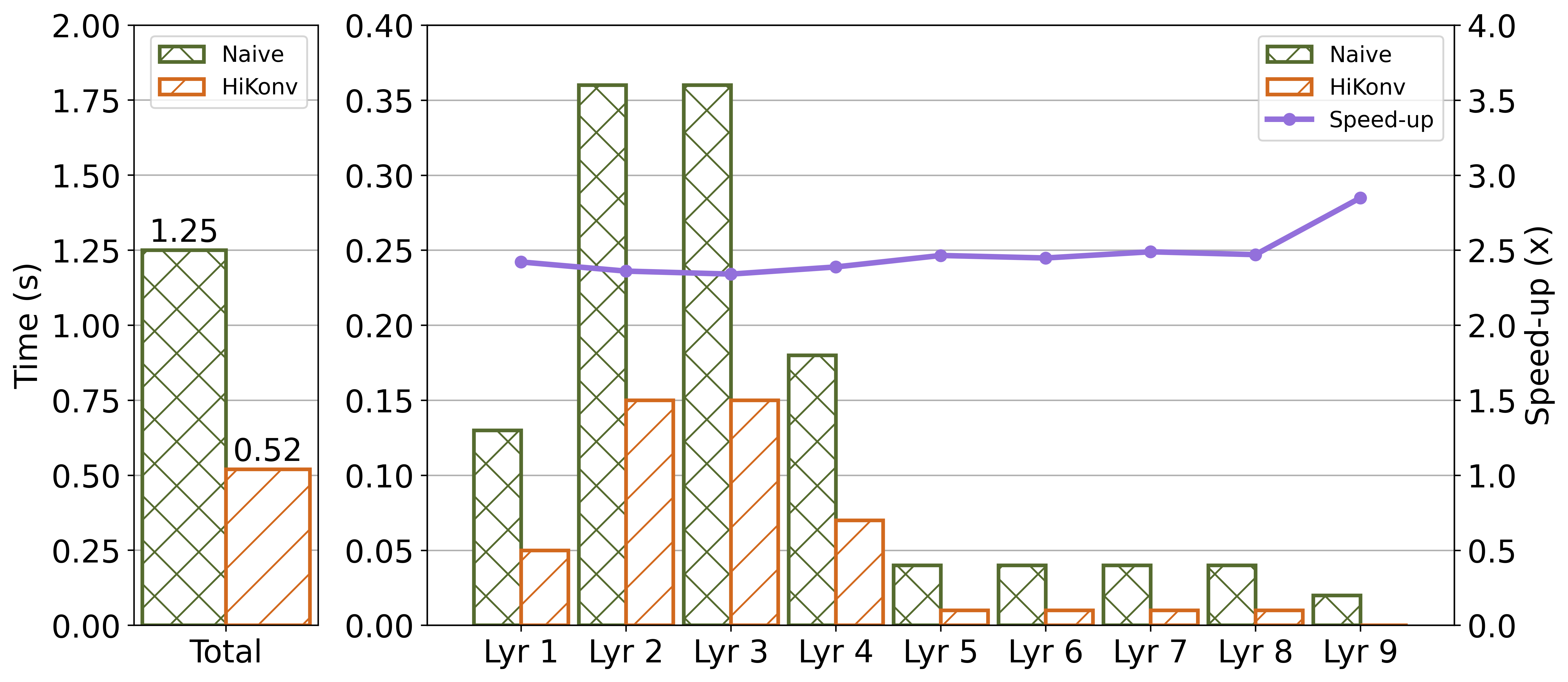}
        \caption{Ultranet on X86\_64 CPU.}
        \label{fig:ultranet_cpu}
    \end{subfigure}\hspace{1em}%
    \begin{subfigure}{0.45\textwidth}
        \centering
        \includegraphics[width=\linewidth]{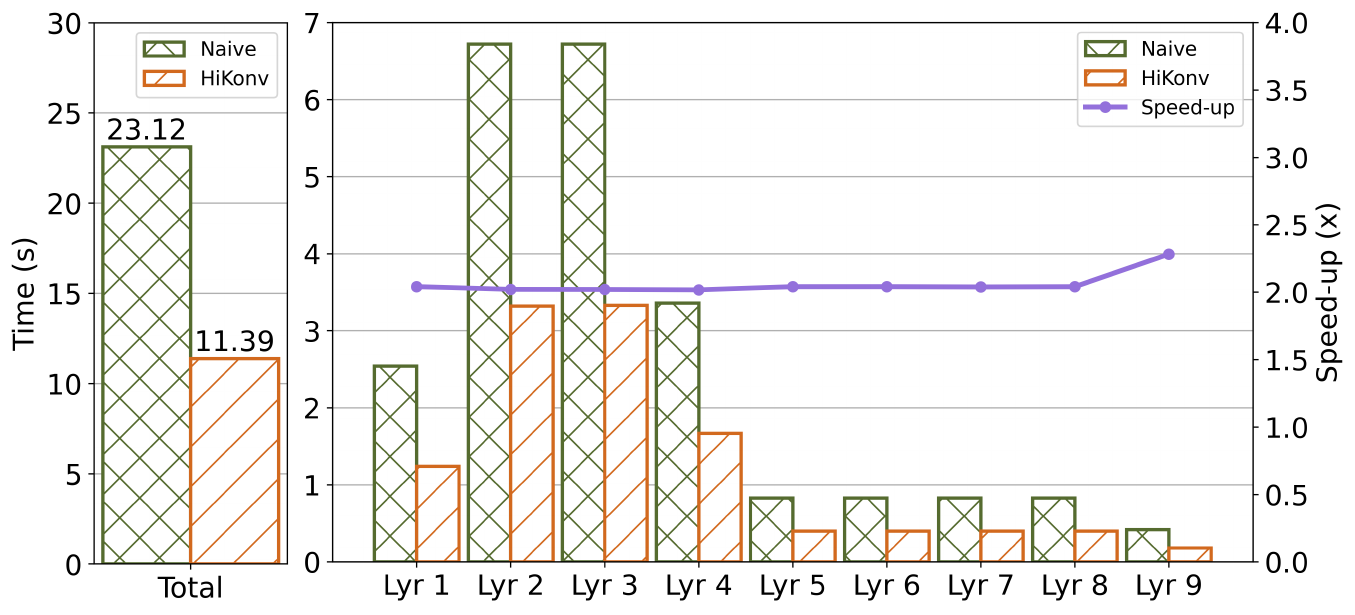}
        \caption{Ultranet on ARM processor.}
        \label{fig:ultranet_arm}
    \end{subfigure}
    \caption{4-bit Ultranet full model evaluation.}
    \vspace{-4mm}
        \label{fig:ultranet_fullmodel}
\end{figure}

To demonstrate the performance of HiKonv in the DNN context, we implement the entire UltraNet with convolution layers based on the HiKonv solution to evaluate on GPPs. In this experiment, we also perform the calculations in integers for both naive and HiKonv implementations to be consistent with the previous experiments and ensure a fair comparison.

The latency results of each layer and the entire model for a 4-bit Ultranet model on the X86\_64 CPU and ARM processor are graphed in Figure~\ref{fig:ultranet_fullmodel}. The speedup gained from HiKonv optimization is consistently over $2\times$, $2.4\times$ for the Intel CPU, and $2.03\times$ for the ARM processor. Compared to the naive implementations, HiKonv convolutions are approximately three times faster, these layer evaluations also include the overhead of packing the kernel on-the-fly for HiKonv implementations and heavier memory access due to a large amount of intermediate feature data. Moreover, one "layer" in the figure includes convolution layers, activation layers, and pooling layers. HiKonv only accelerates convolutions operations; thus, activation and pooling layers are not optimized by HiKonv.
The layer-wise performance on the ARM processor is worse than the same layers on the X86\_64 CPU due to the smaller cache size and lower performance of the ALU in ARM.
Nevertheless, the overall result still demonstrates the effectiveness of HiKonv in a convolutional neural network. Although the other operations that are not optimized by HiKonv somewhat lessen the speedup compared to theoretical values, the benefit of HiKonv are evidently significant.

\subsection{HiKonv on Reconfigurable Hardware}\label{eva:fpga}
We choose the Xilinx Ultra96 MPSoC platform to conduct the evaluation of our implementation. It is equipped with 360 DSP48E2 units and a quad-core ARM processor.
We adopt the High Level Synthesis (HLS) design method for architecture design with C++ as the input language.

\begin{table*}[]
\small
\centering
\caption{Resource consumption of unit convolver.}
\vspace{-2mm}
\resizebox{1.6\columnwidth}{!}{
\begin{tabular}{c|ccc|ccc|ccc}
\hline \hline
      & \multicolumn{3}{c|}{P6Q6}                                              & \multicolumn{3}{c|}{P4Q4}                                              & \multicolumn{3}{c}{P2Q2}                                              \\ \hline
      & \multicolumn{1}{c|}{HiKonv} & \multicolumn{1}{c|}{Conv-DSP} & Conv-LUT & \multicolumn{1}{c|}{HiKonv} & \multicolumn{1}{c|}{Conv-DSP} & Conv-LUT & \multicolumn{1}{c|}{HiKonv} & \multicolumn{1}{c|}{Conv-DSP} & Conv-LUT \\ \hline
DSP   & \multicolumn{1}{c|}{1}      & \multicolumn{1}{c|}{2}        & 0        & \multicolumn{1}{c|}{1}      & \multicolumn{1}{c|}{4}        & 0        & \multicolumn{1}{c|}{1}      & \multicolumn{1}{c|}{10}       & 0        \\ \hline
LUT   & \multicolumn{1}{c|}{76}     & \multicolumn{1}{c|}{51}       & 141      & \multicolumn{1}{c|}{133}    & \multicolumn{1}{c|}{189}      & 328      & \multicolumn{1}{c|}{182}    & \multicolumn{1}{c|}{332}      & 540      \\ \hline
Cycle & \multicolumn{1}{c|}{2}      & \multicolumn{1}{c|}{2}        & 2        & \multicolumn{1}{c|}{2}      & \multicolumn{1}{c|}{3}        & 3        & \multicolumn{1}{c|}{2}      & \multicolumn{1}{c|}{6}        & 6        \\ \hline \hline
\end{tabular}}
\label{tab:singleconvo}
\end{table*}

\subsubsection{Single DSP convolver}

We first evaluate the convolution with a single DSP from the FPGA. We configure the single DSP convolver to process kernel and feature with HiKonv in different bitwidths, noted as HiKonv convolver. We also configure a ``two for-loop" based module to provide the same convolution output as our HiKonv convolver, and apply loop unrolling to guarantee the design is optimized. The design configured to use DSP resources but without adopting the HiKonv solution is noted as Conv-DSP. We also configure the convolver without adopting HiKonv and DSP resources, noted as Conv-LUT. The DSP and LUT resource utilization data are shown in Table~\ref{tab:singleconvo} together with the processing latency of the generated hardware module counted in clock cycles.

Clearly, compared to the conventional convolvers with DSPs, the HiKonv solution uses a single DSP but performs the same amount of computation, whereas the conventional convolvers cost multiple DSPs.
HiKonv convolver occupies slightly more LUTs than conventional convolver when customized for 6-bit inputs because the \textit{Packing Decrementer} and \textit{Split Incrementer} consume more LUTs; however, when the bitwidth is reduced, the HiKonv convolvers cost even less LUTs than the conventional convolvers.
HiKonv-based convolvers always provide shorter processing latency since the multiplication and addition operations all happen in the single multiplication without the loop for accumulation as in normal convolution. The comparison with Conv-LUT further shows the reduced LUT resource with the effective use of DSP. The saving of LUT becomes more significant when the input bitwidth is smaller.

\subsubsection{Comparison to Binary convolution layer}

We then evaluate the extreme case of quantized convolution, which is binary neural networks (BNN).
A convolutional layer in a BNN takes the binary inputs for both feature maps and kernel weights, processes the convolution between them, and generates the outputs. Note that the outputs may not be binary due to channel-wise accumulation.
We first implement a binary convolution layer with 4-bit outputs without using the DSP resources, denoted as BNN-LUT; we then configure a binary computation module with our HiKonv solution, denoted as BNN-HiKonv. In comparison, we evaluate the resource utilization of these two designs under the same concurrency and the same clock frequency setting, as shown in Table~\ref{tab:bnnEval}.

\begin{table}[h]
\small
    \centering
    \caption{Comparison of Resource util. to binary convolution}
    \vspace{-2mm}
    \resizebox{\columnwidth}{!}{\begin{tabular}{c|c|c|c|c|c|c}
    \hline\hline
    \multicolumn{2}{c|}{\# of Concurrent MACs}  & 336 & 576 & 960 & 1536 & 3072\\\hline
     \multirow{1}{*}{BNN-LUT}& LUT  & 3371 & 4987 & 7764 & 12078 & 23607 \\\hline
    \multirow{3}{*}{BNN-HiKonv} & LUT  & 2672 & 2536 & 3369 & 3587 &  9319\\ \cline{2-7}
                                & DSP   & 16   & 32   & 64   & 128  &256 \\ \cline{2-7}
                                & DSP Thro.&21 &18 &15 &12 & 12 \\ \cline{2-7}
                                & LUT/DSP & 43.7 & 76.6&68.7 & 65.4& 55.8\\ \cline{2-7}
    \hline\hline
    \end{tabular}}
    \label{tab:bnnEval}
\end{table}

Clearly, compared to BNN-LUT, the LUT usage of BNN-HiKonv is reduced. 
However, the throughput for each DSP reduces when the concurrency increases because there is more vertical stacking, and it takes more guard bits when the concurrency increases.
The equivalent number of LUTs replaced by one DSP (LUT/DSP) varies from 43.7 to 76.6 due to the accumulation logic in the convolution operation. 
HiKonv creates opportunities to leverage DSPs for high-throughput BNN (or other low-bitwidth models) convolution computations that would help map a larger BNN with high concurrency into the same FPGA. It can also potentially increase the design's clock frequency since DSPs can run at a higher frequency than LUTs.

\subsubsection{Complete model evaluation}

We apply our HiKonv solution to the entire UltraNet model~\cite{UltraNet2020} on the Xilinx Ultra96 MPSoC FPGA.
The HiKonv convolution is implemented in the programmable logic as an accelerator for convolutions.
The weight and activation of this model are quantized to 4-bit. We execute all the convolution layers on the programmable logic and the other layers on the ARM processor on the FPGA platform.
We follow the same layer architecture and system architecture as the original UltraNet design and only change the computation for convolution with our HiKonv solution.
Besides using DSPs, we also use small adders and shifters constructed by LUTs, taking advantage of the flexible configuration features of the FPGA.

In addition to resource utilization, we also measure the throughput in frame-per-second (fps) and calculate the DSP efficiency in terms of Giga-operations-per-second-per-DSP (Gops/DSP) for comparison as shown in Table~\ref{tab:UltraNetimpl}.
All the testing data is first loaded into the DDR to leverage the full capacity of the accelerators in our evaluations.

\begin{table}[h]
\small
\centering
\caption{UltraNet resource and performance.}
\vspace{-2mm}
\resizebox{\columnwidth}{!}{\begin{tabular}{l|c|c|c|c}
\hline\hline
             & LUT & DSP & fps & DSP Eff. (Gops/DSP) \\ \hline
UltraNet     & 4.3k& 360 &  248 & 0.289 \\ \hline
UltraNet-HiKonv & 4.8k &  327&  401/588 & 0.514/0.753 \\ \hline\hline
\end{tabular}}
\label{tab:UltraNetimpl}
\vspace{-2mm}
\end{table}

UltraNet-HiKonv uses more LUT resources than the original implementation due to the shifting and adding logic; however, it reduces the DSP utilization thanks to the dramatic improvements in the efficiency and the throughput of the DSPs.
The original implementation of the UltraNet uses one DSP for two 4-bit MACs that are natively supported by the synthesis tool. It only achieves 248 fps with a 0.289 Gops/DSP efficiency. 
With our HiKonv solution, the onboard implementation of UltraNet achieves 401 fps with a 0.514 Gops/DSP DSP efficiency. This significant improvement is achieved under the constraint that the software execution on the ARM core is not fast enough to feed the input data to the FPGA accelerator to process, even with our best software optimization of multi-threading and data buffering. If this ARM core bottleneck is removed, the UltraNet-HiKonv accelerator can reach an even higher performance of 588 fps with the DSP efficiency of 0.753 Gops/DSP.
\vspace{-2mm}
\section{Conclusion and discussion}\label{sec:conclusion}

In this paper, we present HiKonv, a general technique with theoretical guarantees for using a single multiplier unit to process multiple low-bitwidth convolution operations in parallel for significantly higher computation throughput with flexible bitwidths.
It is able to support convolutions in DNNs and achieves the highest possible throughput for quantized convolution with novel bitwise management and computation. 
As a demonstration of its general applicability and benefits, we show that HiKonv has achieved $3.17\times$ throughput improvement on CPU and $2.37\times$ and $2.61\times$ throughput and DSP efficiency improvements for the DAC-SDC 2020 champion model on FPGA.
HiKonv suits both software and hardware optimizations and provides new opportunities for future hardware designs for efficient DNN processing.

\section*{Acknowledgments}
This work is supported in part by the IBM-Illinois Center for Cognitive Computing Systems Research (C3SR), the National Research Foundation, Prime Minister's Office, Singapore under its Campus for Research Excellence and Technological Enterprise (CREATE) programme, and Xilinx Adaptive Compute Cluster at University of Illinois at Urbana-Champaign.

\bibliographystyle{IEEEtran}
\bibliography{IEEEabrv, ref_full}

\begin{IEEEbiography}[{\includegraphics[width=1in,height=1.25in,clip]{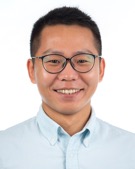}}]{Yao Chen} received the B.S. and Ph.D. degree from Nankai University, Tianjin, China in 2010 and 2016, respectively. He is currently a Senior Research Scientist in the Advanced Digital Sciences Center, Singapore, University of Illinois at Urbana-Champaign. He also serves as the Coordinator for Hardware and Data Analytics research groups. His current research interests include high-performance reconfigurable computing, high-level synthesis, machine learning and Electronic Design Automation.
\end{IEEEbiography}

\begin{IEEEbiography}[{\includegraphics[width=1in,height=1.25in,clip]{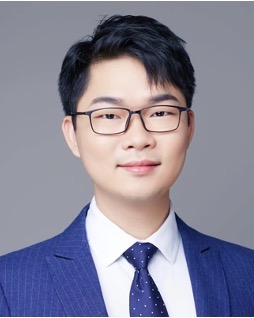}}]{Junhao Pan} received the B.S. degree and is a Ph. D. candidate with the ECE Department, University of Illinois at Urbana-Champaign. His research focuses on deep learning accelerator design, hardware-software co-design, and Internet-of-Things applications.
\end{IEEEbiography}

\begin{IEEEbiography}[{\includegraphics[width=1in,height=1.25in,clip]{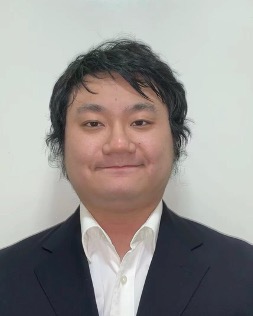}}]{Xinheng Liu}
receives his B.S. and Ph.D degree from University of Illinois Urbana-Champaign in 2015 and 2022 respectively. His research interests include high-level sysnthesis, high-performance computing architecture and domain-specific accelerator design for machine-learning applications.
\end{IEEEbiography}

\begin{IEEEbiography}[{\includegraphics[width=1in,height=1.25in,clip]{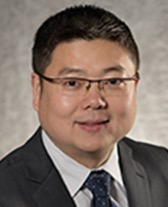}}]{Jinjun Xiong}
Jinjun Xiong received his B.S. and M.S. degrees from Tsinghua University, Beijing, China, in 1998 and 2000, respectively, M.S. degree from University of Wisconsin, Madison in 2002, and Ph.D. degree from University of California, Los Angeles, in 2006. He is currently Empire Innovation Professor with the Department of Computer Science and Engineering at University at Buffalo. Prior to that, he was a Senior Researcher and Program Director for AI and Hybrid Clouds Systems at the IBM Thomas J. Watson Research Center, Yorktown Heights, New York, USA. His current research interests are on across-stack AI systems research, which include AI applications, algorithms, tooling, and computer architectures. Many of his research results have been adopted in industrial products and tools. His publication won eight Best Paper Awards and eight Nominations for Best Paper Awards from international conferences.
\end{IEEEbiography}

\begin{IEEEbiography}[{\includegraphics[width=1in,height=1.25in,clip]{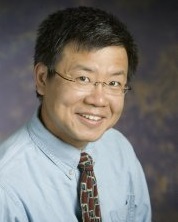}}]{Deming Chen}
received the B.S. degree in computer science from the University of
Pittsburgh, PA, USA, in 1995, and the M.S. and Ph.D. degrees in computer science
from the University of California at Los Angeles, in 2001 and 2005, respectively. He
is the Abel Bliss Endowed Professor 
with the ECE Department, University of Illinois at Urbana–Champaign. Dr. Chen is an IEEE Fellow, an ACM Distinguished Speaker, and the Editor-in-Chief of ACM Transactions on Reconfigurable Technology and Systems (TRETS).
His current research interests include system-level and high-level synthesis, machine learning, GPU and reconfigurable computing, computational genomics, and hardware security.
\end{IEEEbiography}


\end{document}